\documentclass[a4paper]{article}
\pdfoutput=1
\usepackage{microtype}%if unwanted, comment out or use option "draft"
\usepackage{xspace,hyperref,fullpage}
\usepackage{xcolor}
\usepackage{subcaption}
\usepackage{amsthm,amssymb,amsmath,graphics,epsfig,url}

\theoremstyle{plain}
\newtheorem{theorem}{Theorem}
\newtheorem{lemma}{Lemma}
\newtheorem{corollary}{Corollary}
\newtheorem{claim}{Claim}
\newtheorem{fact}{Fact}

\theoremstyle{remark}
\newtheorem{observation}{Observation}

\theoremstyle{definition}
\newtheorem{definition}{Definition}

\newcommand{\pre}{\text{pre}}
\newcommand{\suf}{\text{suf}}
\newcommand {\ov}{\text{ov}}
\newcommand{\setof}[2]{\{#1\mid#2\}}

\newcommand{\N}{\mathbb{N}\xspace}

\newcommand{\Oh}{\mathcal{O}}

\title{Bipartite Graphs of Small Readability} %\footnote{This work was partially supported by someone.}}

\author{Rayan Chikhi\thanks{CNRS, UMR 9189. {\tt rayan.chikhi@univ-lille1.fr}} \and Vladan Jovi\v{c}i\'c\thanks{ENS Lyon, France. {\tt vladan94.jovicic@gmail.com}} \and Stefan Kratsch\thanks{Institut f\"{u}r Informatik, Humboldt-Universit\"{a}t zu Berlin. {\tt kratsch@informatik.hu-berlin.de}} \and Paul Medvedev\thanks{The Pennsylvania State University, USA. {\tt pashadag@cse.psu.edu}} \and Martin Milani\v{c}\thanks{IAM and FAMNIT, University of Primorska, Koper, Slovenia. {\tt martin.milanic@upr.si}} \and Sofya Raskhodnikova\thanks{Boston University, USA. {\tt \{sofya, nvarma\}@bu.edu}} \and Nithin Varma\footnotemark[6]}

\begin{document}

\maketitle

\begin{abstract}
We study a parameter of bipartite graphs called readability, introduced by Chikhi et al.\ ({\em Discrete Applied Mathematics}, 2016) and motivated by applications of overlap graphs in bioinformatics. The behavior of the parameter is poorly understood. The complexity of computing it is open and it is not known whether the decision version of the problem is in NP. The only known upper bound on the readability of a bipartite graph (following from a work of Braga and Meidanis, {\em LATIN} 2002) is exponential in the maximum degree of the graph.

Graphs that arise in bioinformatic applications have low readability.
In this paper we focus on graph families with readability $o(n)$, where $n$ is the number of vertices. We show that the readability of $n$-vertex bipartite chain graphs is between $\Omega(\log n)$ and $\Oh(\sqrt{n})$. We give an efficiently testable characterization of bipartite graphs of readability at most~2 and completely determine the readability of grids, showing in particular that their readability never exceeds~$3$. As a consequence, we obtain a polynomial time algorithm to determine the readability of induced subgraphs of grids.
One of the highlights of our techniques is the appearance of Euler's totient function in the analysis of the readability of bipartite chain graphs. We also develop a new technique for proving lower bounds on readability, which is applicable to dense graphs with a large number of distinct degrees.
\end{abstract}

\section{Introduction}

In this work we further the study of {\em readability} of bipartite graphs initiated by Chikhi et al.~\cite{ChikhiMMR16}. Given a bipartite graph $G=(V_s,V_p,E)$, an {\em overlap labeling} of $G$ is a mapping from vertices to strings, called labels, such that for all $u \in V_s$ and $v \in V_p$ there is an edge between $u$ and $v$ if and only if the label of $u$ {\em overlaps} with the label of $v$ (i.e., a non-empty suffix of $u$'s label is equal to a prefix of $v$'s label). The {\em length} of an overlap labeling of $G$ is the maximum length (i.e., number of characters) of a label. The {\em readability} of $G$, denoted $r(G)$, is the smallest nonnegative integer $r$ such that there is an overlap labeling of $G$ of length $r$. We emphasize that in this definition, no restriction is placed on the alphabet.
One could also consider variants of readability parameterized the size of the alphabet. A result of Braga and Meidanis~\cite{BM02} implies that these variants are within constant factors of each other, where the constants are logarithmic in the alphabet sizes.

The notion of readability arises in the study of overlap digraphs. 
Overlap digraphs constructed from DNA strings have various applications in bioinformatics.\footnote{In the context of genome assembly, variants of overlap digraphs appear as either de Bruijn graphs~\cite{IW95} or string graphs~\cite{M05,SGA} and are the foundation of most modern assemblers (see \cite{MKS10,NP13} for a survey).
	Several graph-theoretic parameters of overlap digraphs have been studied~\cite{BFKSW02,BFKK02,BHKW99,GP14,LZ07,LZ10,PSW03,TU88},
	with a nice survey in~\cite{K16}.} 
Most of the graphs that occur as the overlap graphs of genomes have low readability. 
Chikhi et al.~\cite{ChikhiMMR16} show that the readability of overlap digraphs is asymptotically equivalent to that of balanced bipartite graphs: there is a bijection between overlap digraphs and balanced bipartite graphs that preserves readability up to (roughly) a factor of $2$.
%But inversely, whether low-readability graphs have any properties or structure is currently unknown.
This motivates the study of bipartite graphs with low readability. In this work we derive several results about bipartite graphs with readability sublinear in the number of vertices.

For general bipartite graphs, the only known upper bound on readability is implicit in a paper on overlap digraphs by Braga and Meidanis~\cite{BM02}. As observed by Chikhi et al.~\cite{ChikhiMMR16},  it follows from~\cite{BM02} that the readability of a bipartite graph is well defined and at most~$2^{\Delta+1}-1$, where $\Delta$ is the maximum degree of the graph. Chikhi et al.~\cite{ChikhiMMR16} showed that almost all bipartite graphs with $n$ vertices in each part have readability $\Omega(n/\log n)$. They also constructed an explicit graph family (called Hadamard graphs) with readability $\Omega(n)$.

For trees, readability can be defined in terms of an integer function on the edges, without any reference to strings or their overlaps~\cite{ChikhiMMR16}.
In this work, we reveal another connection to number theory, through Euler's totient function, and use it to prove an upper bound on the readability of bipartite chain graphs.

So far, our understanding of readability has been hindered by the difficulty of proving lower bounds. Chikhi et al.~\cite{ChikhiMMR16} developed a lower bound technique for graphs where the overlap between the neighborhoods of any two vertices is limited. In this work, we add another technique to the toolbox. Our technique is applicable to dense graphs with a large number of distinct degrees. We apply this technique to obtain a lower bound on readability of  bipartite chain graphs.

We give a characterization of bipartite graphs of readability at most~$2$ and use this characterization to obtain a polynomial time algorithm for checking if a graph has readability at most~$2$. This is the first nontrivial result of this kind: graphs of readability at most~$1$ are extremely simple (disjoint unions of complete bipartite graphs, see~\cite{ChikhiMMR16}), whereas the problem of recognizing graphs of readability~3 is open.

We also give a formula for the readability of grids, showing in particular that their readability never exceeds~$3$. As a corollary, we obtain a polynomial time algorithm to determine the readability of induced subgraphs of grids.

\subsection{Our Results and Structure of the Paper}

Preliminaries are summarized in Section~\ref{sec:def}; here we only state some of the most important technical facts. 
In the study of readability, it suffices to consider bipartite graphs that are connected and {\em twin-free}. A bipartite graph is twin-free if no two vertices in the same part have the same sets of neighbors~\cite{ChikhiMMR16}.
%The readabilities of digraphs and bipartite graphs are asymptotically equivalent~\cite{ChikhiMMR16}.
Since connected bipartite graphs have a unique bipartition up to swapping the two parts, some of our results are stated without specifying the bipartition.
\paragraph{Bounds on the readability of bipartite chain graphs (Section~\ref{sec:cnn-bounds}).}
Bipartite chain graphs are the bipartite analogue of a family of digraphs that occur naturally as subgraphs of overlap graphs of genomes.
A {\it bipartite chain graph} is a bipartite graph $G = (V_s,V_p,E)$ such that the vertices in $V_s$ (or $V_p$) can be linearly ordered with respect to inclusion of their neighborhoods.
That is,
we can write $V_s = \{v_1,\ldots, v_k\}$ so that $N(v_1)\subseteq \ldots\subseteq N(v_k)$
(where $N(u)$ denotes the set of $u$'s neighbors).
%In a bipartite chain graph, for every two vertices $u$ and $v$ from the same part, we have $N(u)\subseteq N(v)$ or $N(v)\subseteq N(u)$.
A twin-free connected bipartite chain graph must have the same number of vertices on either side. For each $n \in \N$, there is, up to isomorphism, a unique connected twin-free
bipartite chain graph with $n$ vertices in each part, denoted $C_{n,n}$.
%We can define this graph as %refer to the twin-free bipartite chain graph with $n+n$ vertices.  Formally,
The graph $C_{n,n}$ is $(V_s,V_p,E)$ where $V_s =\{s_1,\ldots,s_n\},
V_p=\{p_1,\ldots,p_n\}$, and $E=\{(s_i,p_j) \mid 1\le i\leq j\le n\}$. The graph $C_{4,4}$ is shown in Figure~\ref{fig:chainGraph}.
We prove an upper and a lower bound on the readability of $C_{n,n}$.
\begin{figure}[t]%[h!]
	\begin{center}
		\includegraphics[width=0.25\linewidth]{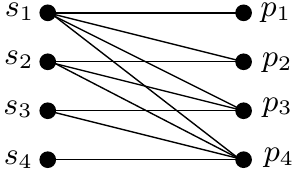}
	\end{center}
	\caption{The graph $C_{4,4}$}\label{fig:chainGraph}
\end{figure}
\begin{theorem}\label{theorem:ub-bip-chain-graphs}
	For all $n \in \N$, the graph $C_{n,n}$ has readability $\Oh(\sqrt{n})$, with labels over an alphabet of size 3.
\end{theorem}

We prove Theorem~\ref{theorem:ub-bip-chain-graphs} by giving an efficient
%\snote{We changed this from polynomial time to efficient, as it was not clear what was the parameter on which the running time depended. We could probably get into the details of the running time in the journal version of the paper.}
algorithm that constructs an overlap labeling of $C_{n,n}$ of length $\Oh(\sqrt{n})$ using strings over an alphabet of size~$3$.

\begin{theorem}\label{theorem:lb-bip-chain-graphs}
	For all $n \in \N$, the graph $C_{n,n}$ has readability $\Omega(\log n)$.
\end{theorem}

\paragraph{Characterization of bipartite graphs with readability at most~$2$ (Section~\ref{sec:readability-2-characterization}).}
Let $C_t$ for $t \in \N$ denote the simple cycle with $t$ vertices. The {\em domino} is the graph obtained from the cycle $C_6$ by adding an edge between two diametrically opposite vertices. For a graph $G$ and a set $U \subseteq V(G)$, let $G[U]$ denote the subgraph of $G$ induced by $U$.

%\mnote{I think we should replace $F$ in all such contexts with some letter from the end of the alphabet, as it denotes a set of vertices and not a graph. (In graph theory, letter $F$ is commonly used for some specific or forbidden graph.) For example, $G[U]$ is common and $U$ is never used in the paper, so it is safe. \rayan{RC: yes that sounds good, and in fact $U$ is used later in the paper for subgraphs -- $U$ would work (or $W$ too).} }\nvnote{Addressed.}

Chikhi et al.~\cite{ChikhiMMR16} proved that every bipartite graph with readability at most~$1$ is a disjoint union of complete bipartite graphs (also called bicliques).
The characterization in the following theorem
%, which we prove in Section~\ref{sec:readability-2-characterization},
extends our understanding to graphs of readability at most~$2$. Recall that a {\em matching} in a graph is a set of pairwise disjoint edges.

\begin{theorem}\label{thm:readability-2-characterization}
	A twin-free bipartite graph $G$ has readability at most~$2$ if and only if~$G$ has a matching $M$ such that the graph $G' = G - M$ satisfies the following properties:
	\begin{enumerate}
		\item $G'$ is a disjoint union of complete bipartite graphs.
		\item For $U\subseteq V(G)$, if $G[U]$ is a $C_6$, then $G'[U]$ is the disjoint union of three edges.
		\item For $U\subseteq V(G)$, if $G[U]$ is a domino, then $G'[U]$ is the disjoint union of a $C_4$ and an edge.
	\end{enumerate}
\end{theorem}

Note that Theorem~\ref{thm:readability-2-characterization} expresses a condition on vertex labels of a bipartite graph in purely graph theoretic terms. This reduces the problem of deciding if a graph has readability at most~$2$ to checking the existence of a matching with a specific property.

\paragraph{An efficient algorithm for readability $2$ (Section~\ref{sec:readability-2}).}
It is unknown whether computing the readability of a given bipartite graph is NP-hard.
In fact, it is not even known whether the decision version of the problem is in NP,
as the only upper bound on the readability of a bipartite graph with $n$ vertices in each part is $\Oh(2^n)$~\cite{BM02}.
%We make a first step on this front by giving a polynomial time algorithm to determine whether a given bipartite graph has readability at most~$2$.
%\mnote{I do not understand the purpose of this sentence and the next one. It is still unknown whether the decision version of the readability problem is in NP, so why saying ``it was previously unknown''? Also, I do not see how our algorithm for readability 2 can be seen as ``a first step on this front''.\paul{PM: What about now? We basically answer the question and more for readability two.}}
We make progress on this front by showing that for readability 2, the decision version is polynomial time solvable.
\begin{theorem}\label{thm:readability-2}
There exists an algorithm that, given a bipartite graph $G$, decides in polynomial time whether $G$ has readability at most~$2$.
\end{theorem}
Moreover, if the answer is ``yes", the algorithm can also produce an overlap labeling of length at most~$2$.

\paragraph{Readability of grids and grid graphs (Section~\ref{sec:grids}).}
We give a full characterization of the readability of grids. A {\it (two-dimensional) grid} is a graph $G_{m,n}$ with vertex set $\{0,1,\dots,m-1\}\times \{0,1,\dots,n-1\}$ such that there is an edge between two vertices
if and only if~the $L_1$-distance between them is~$1$.
%$G_{m,n}$ is bipartite because it does not contain odd-length cycles.
An example is shown in Figure~\ref{fig:grid_graph_4x4}.
The following theorem fully settles the question of readability of grids.
\begin{figure}[t]%[h!]
	\begin{center}
		\includegraphics[width=0.65\linewidth]{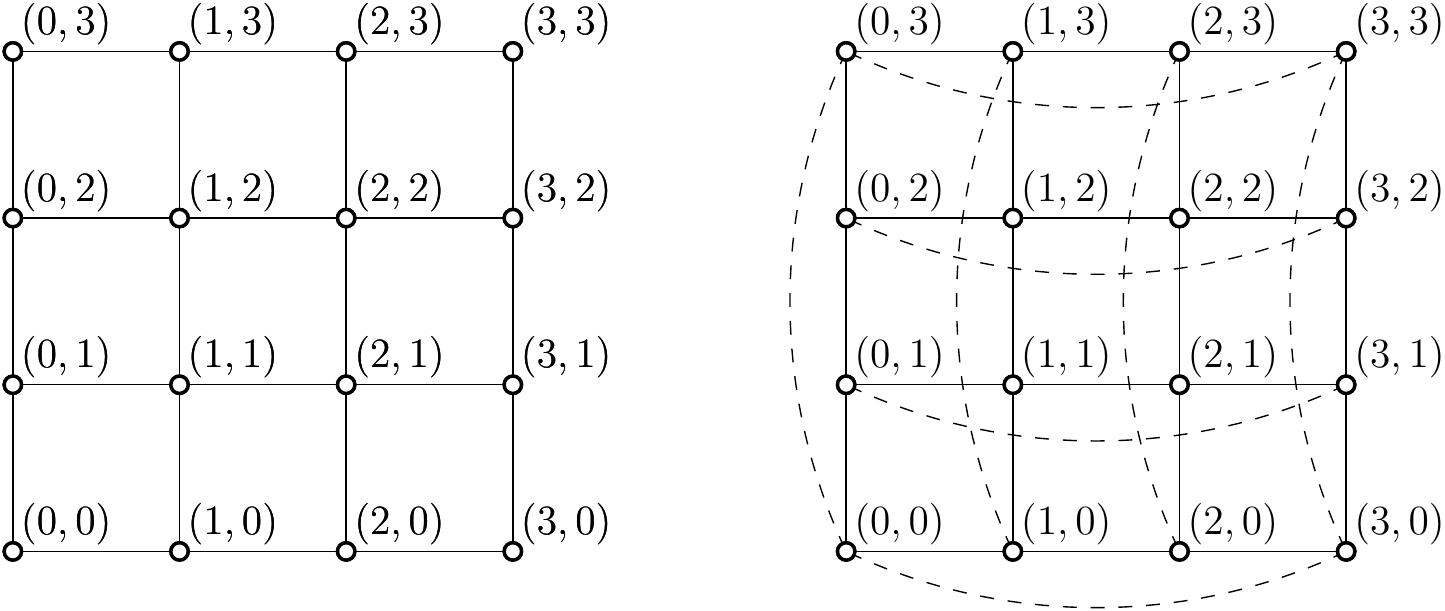}
	\end{center}
	\caption{The $4\times 4$ grid $G_{4,4}$ and toroidal grid ${\it TG}_{4,4}$.}\label{fig:grid_graph_4x4}
\end{figure}
%
%if and only if  $\|u-v\|_1=1$, where $\|u-v\|_1$ denotes the $L_1$ distance between $u$ and $v$.
\begin{theorem}\label{thm:rgrid}
	For any two positive integers $m, n$ with $m\le n$, we have
	$$r(G_{m,n})= \left\{
	\begin{array}{ll}
	3, & \hbox{if $m\ge 3$ and $n\ge 3$;} \\
	2, & \hbox{if $(m = 2$ and $n\ge 3)$ or $(m = 1$ and $n\ge 4)$;} \\
	1, & \hbox{if $(m,n)\in \{(1,2),(1,3),(2,2)\}$;}\\
	0, & \hbox{if $m = n = 1$.}
	\end{array}
	\right.$$
\end{theorem}

Theorem~\ref{thm:rgrid} has an algorithmic implication for the readability of grid graphs, where a {\it grid graph} is an induced subgraph of a grid. Several problems are known to be NP-hard on the class of grid graphs, including Hamiltonicity problems~\cite{IPS82}, various layout problems~\cite{DPPS01}, and others (see, e.g.,~\cite{CCJ90}). We show that unless P = NP, this is not the case for the readability problem.

\begin{corollary}\label{cor:grid-graphs}
	The readability of a given grid graph can be computed in polynomial time.
\end{corollary}
%Finally, we note that all of our results have equivalents in the directed graph model of readability,
%though we do not explicitly state them. All missing proofs are included in the appendix.
\subsection{Technical Overview}\label{sec:techover}

We now give a brief  description of our techniques. The key to proving the upper bound on the readability of bipartite chain graphs is understanding the combinatorics of the following process.
We start with the sequence $(1,2)$.
%Sofya: It is ok to omit parentheses here. We just can't write things like B=1,2
The process consists of a series of rounds, and as a convention, we start at round 3: we write $3$ ($=1+2$)  between $1$ and $2$ and obtain the sequence $(1,3,2)$.
%Sofya: Let's keep this clean, even if it does not correspond to the technical section. However, Paul, I think it would be nicer if you went to back to the cleaner view in the technical section, too.
%Also note that a pair is one object. ``Insert between (a,b)'' does not parse.
 %\pnote{I changed this section to better agree with what we later have in the main body: in particular, the rounds starts with round 2, and we deal with the non-symmetric version of the problem. \rayan{RC: I further slightly changed the presentation, and it seems that we actually start at round 3.}}
%We start with the numbers $1$ and $2$ written next to each other in the first round. In the next round we write $3$ ($=1+2$) in between $1$ and $2$ and obtain the sequence $1,3,2$.
More generally, in round $r$, we insert $r$ between all the consecutive pairs of numbers in the current sequence that sum up to $r$. Thus, we obtain $(1,4,3,2)$ in round 4, then $(1,5,4,3,5,2)$ in round 5, and so on. The question is to determine the length of the sequence formed in round $r$ as a function of $r$.
We prove that this length is $\frac 12\sum^r_{k=1} \varphi(k) = \Theta(r^2)$, where $\varphi(k)$ is the famous Euler's totient function denoting the number of integers in $\{1,\dots,k\}$ that are coprime to $k$.
%We think that the appearance of number-theoretic notions in the study of readability is surprising and of independent interest.

To prove our lower bound on the readability of
%Sofya: I don't think we need ``twin-free'' detail here. We did not even explain the quantifiers in our lower bound, so it does not make sense to get into technical details.
%twin-free
bipartite chain graphs, we define a special sequence of subgraphs of the bipartite chain graph such that the number of graphs in the sequence is a lower bound on the readability. The sequence that we define has the additional property that if two vertices in the same part have the same set of neighbors in one of the graphs, then they have the same set of neighbors in all of the preceding graphs in the sequence. If the readability is very small, then we cannot simultaneously cover all the edges incident with two large-degree nodes as well as have their degrees distinct.
%Here, we are also using the fact that all vertices in the same part of a bipartite chain graph have distinct degrees.
The only properties of the connected twin-free bipartite chain graph that our proof uses are that it is dense and all vertices in the same part have distinct degrees.
Hence, this technique is more broadly applicable to any graph class satisfying these properties.

Our characterization of graphs of readability at most~$2$, roughly speaking, states that a twin-free bipartite graph has readability at most~$2$ if and only if the graph can be decomposed into two subgraphs $G_1$ and $G_2$ such that $G_1$ is a disjoint union of bicliques and $G_2$ is a matching satisfying some additional properties. For $i\in \{1,2\}$, the edges in $G_i$ model overlaps of length exactly $i$. The heart of the proof lies in observing that for each pair of bicliques in the first subgraph, there can be at most~one matching edge in the second subgraph that has its left endpoint in the first biclique and the right endpoint in the second biclique.

To derive a polynomial time algorithm for recognizing graphs of readability two, we first reduce the problem to connected twin-free graphs of maximum degree at least three. For such graphs, we show that the constraints from our characterization of graphs of readability at most~$2$ can be expressed with a 2SAT formula having variables on edges and modeling the selection of edges forming a matching to form the graph $G_2$ of the decomposition.

In order to determine the readability of grids, we establish upper and lower bounds and in both cases use the fact that readability is monotone under induced subgraphs (that is, the readability of a graph is at least the readability of each of its induced subgraphs).
The upper bound is derived by observing that every grid is an induced subgraph of some $4n\times 4n$ toroidal grid (see Figure~\ref{fig:grid_graph_4x4}) and exploiting the symmetric structure of such toroidal grids to show that their readability is at most~$3$. This is the most interesting part of our proof and involves partitioning the edges of a $4n\times 4n$ toroidal grid into three sets and coming up with labels of length at most~$3$ for each vertex based on the containment of the four edges incident with the vertex in each of these three parts. Our characterization of graphs of readability at most~$2$ is a helpful ingredient in proving the lower bound on the readability of grids, where we construct a small subgraph of the grid for which our characterization easily implies that its readability is at least $3$.

\section{Preliminaries}\label{sec:def}

For a string $x$, let $\pre_i(x)$ (respectively, $\suf_i(x)$) denote the prefix (respectively, suffix) of $x$ of length $i$. A string $x$ {\em  overlaps} another string $y$ if there exists an $i$ with $1 \leq i \le \min\{|x|,|y|\}$ such that $\suf_i(x) = \pre_i(y)$. If $1 \le i < \min\{|x|,|y|\}$, we say that $x$ {\em properly overlaps} with $y$.
%Let $y$ be another string. We denote by $x\cdot y$ the concatenation of $x$ and $y$.
%We say that $x$ {\em  overlaps} $y$ if there exists an $i$ with $1 \leq i \le \min\{|x|,|y|\}$ such that $\suf_i(x) = \pre_i(y)$. In this case, we say that $x$ overlaps $y$ by $i$ (or that $x$ and $y$ have an overlap of length $i$).
%Define $\ov(x,y)$ as the minimum $i$ such that $x$ overlaps $y$ by $i$, or $0$ if $x$ does not overlap $y$.
%The overlap between is {\em proper} if $1 \le i < \min\{|x|,|y|\}$.
%An overlap between $x$ and another string $y$ is {\em proper} if its length is less than $\min\{|x|,|y|\}$.
For a positive integer $k$, we denote by $[k]$ the set $\{1,\ldots, k\}$.
Let $G= (V,E)$ be a (finite, simple, undirected) graph.
%For a vertex $v$ in a graph, we denote the set of neighbors of $v$ by $N(v)$.
%An {\em independent set} in $G$ is a set of pairwise non-adjacent vertices.
%A graph is said to be {\em bipartite} if it has a {\it bipartition}, that is, a partition of its vertex set into two (possibly empty) independent sets.
If $G$ is a connected bipartite graph, then it has a unique bipartition (up to the order of the parts).
%and the two sets are called {\em parts}.
In this paper, we consider bipartite graphs $G = (V,E)$. If the bipartition $V= V_s\cup V_p$ is specified, we denote such graphs by $G = (V_s, V_p, E)$.
%Sofya: Before, it said that we consider graphs with a fixed bipartition when, in fact, we frequently do not specify the bipartition.
Edges of a bipartite graph $G$ are denoted by $\{u,v\}$ or by $(u,v)$ (which implicitly implies that $u\in V_s$ and $v\in V_p$).
We respect bipartitions when we perform graph operations such as taking an induced subgraph and disjoint union.
For example, we say that a bipartite graph $G_1 = (V_s^1,V_p^1,E_1)$ is an {\it induced subgraph} of a bipartite graph $G_2 = (V_s^2,V_p^2,E_2)$ if $V_s^1\subseteq V_s^2$, $V_p^1\subseteq V_p^2$, and $E_1 = E_2\cap \{(x,y): x\in V_s^1, y\in V_p^1\}$.
The {\it disjoint union} of two vertex-disjoint bipartite graphs $G_1 = (V_s^1,V_p^1,E_1)$ and
$G_2 = (V_s^2,V_p^2,E_2)$ is the bipartite graph $(V_s^1\cup V_s^2,V_p^1\cup V_p^2,E_1\cup E_2)$.

%A {\em biclique} is a complete bipartite graph. Note that the one-vertex graph is a biclique (with one of the parts of the bipartition being empty).
The path on $n$ vertices is denoted by $P_n$. Given two graphs $F$ and $G$, graph $G$ is said to be {\em $F$-free} if no induced subgraph of $G$ is isomorphic to $F$.
%\nithin{A {\em  matching} in a graph is a set of pairwise disjoint edges. A matching in a graph is {\em perfect} if every vertex of the graph is incident to an edge of the matching. }
Two vertices $u,v$ in a bipartite graph
are called {\em twins} if they belong to the same part of the bipartition and
have the same neighbors (that is, if $N(u)=N(v)$).
Given a bipartite graph $G = (V_s,V_p,E)$ we can define its {\it twin-free reduction} ${\it TF}(G)$
as the graph with vertices being the equivalence classes of the twin relation on $V(G)$ (that is, $x\sim y$ if and only if $x$ and $y$ are twins in $G$), and two classes $X$ and $Y$ are adjacent if and only if $(x,y)\in E$ for some $x\in X$ and $y\in Y$. For graph theoretic terms not defined here, we refer to~\cite{MR1367739}.

We now state some basic results for later use. %The proofs of these are deferred to the appendix. %for later use.
\begin{lemma}\label{lem:induced-subgraph}
Let $G$ and $H$ be two bipartite graphs.
\begin{enumerate}
 \item[(a)] If $G$ is an induced subgraph of $H$, then $r(G) \le r(H)$.
 \item[(b)] If $F$ is the disjoint union of $G$ and $H$, then
		$r(F) = \max\{r(G), r(H)\}$.
 \item[(c)] The readability of $G$ is the same for all bipartitions of $V(G)$.
 \item[(d)] $r(G) = r({\it TF}(G))$.
\end{enumerate}
\end{lemma}
% % \begin{lemma}\label{lem:induced-subgraph}
% % Let $G$ and $H$ be two bipartite graphs.\\
% % {\em\bf (a)} If $G$ is an induced subgraph of $H$, then $r(G) \le r(H)$.\\
% % 	%If $G$ and $H$ are bipartite graphs such that $G$ is an induced subgraph of $H$, then $r(G) \le r(H)$.
% % %\end{lemma}
% % %\begin{restatable}{replemma}{lemdisjointunion}\label{lem:disjoint-union}
% % {\em\bf (b)} If $F$ is the disjoint union of $G$ and $H$ then
% % 		$r(F) = \max\{r(G), r(H)\}$.\\
% % %\end{restatable}
% % %Lemma~\ref{lem:induced-subgraph}{(\bf b)} shows that the study of readability reduces to the case of connected bipartite graphs.
% % %Note also that a connected bipartite graph has a unique bipartition up to switching the two parts;
% % %in this case, a simple string reversing argument shows that the value of the readability is independent of the choice of the %bipartition.
% % %\noindent Next we show that the readability of a bipartite graph does not depend on the bipartition.
% % %\begin{restatable}{replemma}{lembipartitiondoesnotmatter}\label{lem:bipartition-does-not-matter}
% % {\em\bf (c)} The readability of $G$ is the same for all bipartitions of $V(G)$.\\
% % 	%The readability of a bipartite graph $G$ is the same for all bipartitions of $V(G)$.
% % {\em\bf (d)} $r(G) = r({\it TF}(G))$.
% % \end{lemma}
\begin{proof}
	$(a)$	If $\ell$ is any overlap labeling for $H$ then the restriction of $\ell$ to $V(G)$ yields an overlap labeling for $G$. Thus,
	$r(G) \le r(H)$.
	
	%\lemdisjointunion*
\begin{sloppypar}
	$(b)$	Part $(a)$ implies that
	$r(G)\le r(F)$  and
	$r(H)\le r(F)$; thus $r(F)\geq\max\{r(G),r(H)\}$.
	%$r(G_i)\le r(G)$ for $i\in \{1,2\}$.
	On the other hand, let $\ell_G$ and $\ell_H$ be optimal labelings of $G$ and
	$H$, over $\Sigma_G$ and $\Sigma_H$, respectively.
	By introducing new characters if necessary, we may assume that $\Sigma_G\cap
	\Sigma_H = \emptyset$.
	Thus, the combined labeling $\ell$ of $F$ over $\Sigma = \Sigma_G\cup
	\Sigma_H$, defined as
	$$\ell(x) = \left\{
	\begin{array}{ll}
	\ell_G(x), & \hbox{if $x\in V(G)$;} \\
	\ell_H(x), & \hbox{if $x\in V(H)$.}
	\end{array}
	\right.$$
	for all $x\in V(F)$, is an overlap labeling of $F$, showing that
	$r(F) \le \max\{r(G), r(H)\}$.
\end{sloppypar}
	
	%\lembipartitiondoesnotmatter*
	$(c)$	By part $(b)$, the readability of $G$ is the maximum readability of a connected component of $G$. Therefore, it is sufficient to prove the lemma for the case when $G$ is connected.
	Every connected graph has a unique bipartition, up to switching the roles of $V_s$ and $V_t$. Switching the roles of $V_s$ and $V_t$ in a graph does not affect its readability, because an overlap labeling of the new graph can be obtained by reversing all the labels in the overlap labeling of the original graph. Thus, the readability of $G$ is not affected by the choice of bipartition of $V(G)$.
	
	%\lemmatwinfreereduction*
	$(d)$  It suffices to prove that for a pair of twins $u$ and $v$, $r(G) =r(G-u)$.
	By part $(a)$, we have $r(G-u) \le r(G)$.
	Conversely, an optimal overlap labeling $\ell$ of $G-u$ can be extended to an overlap labeling $\ell'$ of $G$
	of the same maximum length as $\ell$ by setting, for all $x\in V(G)$,
	$$\ell'(x) = \left\{
	\begin{array}{ll}
	\ell(x), & \hbox{if $x\in V(G)\setminus\{v\}$;} \\
	\ell(u), & \hbox{if $x = v$.}
	\end{array}
	\right.$$
	Thus,  $r(G) \le r(G-u)$.
\end{proof}
%In particular, this implies that $r(G) \le r(H)$ whenever $G$ and $H$ are bipartite graphs (without a specified bipartition)
%and $G$ is an induced subgraph of $H$ in the usual graph theoretic sense\rcnote{This sentence.. feels strikingly similar to Lemma 2.1 :)}.
Lemma~\ref{lem:induced-subgraph}{(b)} shows that the study of readability reduces to the case of connected bipartite graphs.
By Lemma~\ref{lem:induced-subgraph}{(c)}, the readability of a bipartite graph is well defined even if a bipartition is not given in advance.
We state our results without specifying a bipartition in Sections~\ref{sec:readability-2-characterization}-\ref{sec:readability-2}.
%, \ref{sec:grids}, and~.
Lemma~\ref{lem:induced-subgraph}{(d)} further shows that to understand the readability of connected bipartite graphs, it suffices to study the readability of connected twin-free bipartite graphs.
%A further reduction is shown in the next lemma.
%Note that a twin-free graph can have up to one isolated vertex in each part of the bipartition.
%}
%\nvnote{The entire premise of this paragraph is that we can restrict our attention to the connected case. Then, why is it relevant to pay attention to isolated vertices and defining twin free reduction that respects the bipartition?}
%\snote{I agree with Nithin. We need to move this definition. Also, I think we need isolated vertices to not be twins for the purposes of HUB decomposition. We can just take the definition from [CMMR].}

%\begin{restatable}{replemma}{lemmatwinfreereduction}\label{lem:twin-free-reduction}
%	Gg
%end{restatable}

 %, where a bipartite graph is {\it twin-free} if it has no pairs of twins.
%(That is, if for $x,y\in V_s$ or $x,y\in V_p$, condition $N(x) = N(y)$ implies $x = y$.)
%For later use, we also recall the following simple observation (which follows, e.g., from~\cite[Theorem 4.3]{ChikhiMMR16}).
%By $P_4$ we denote the path on $4$ vertices.

%\begin{lemma}\label{lem:small-readability}
%	A bipartite graph $G$ has: (i) $r(G) = 0$ if and only if $G$ is edgeless, and
%	(ii) $r(G) \le 1$ if and only if $G$ is $P_4$-free (equivalently: a disjoint union of bicliques).
%\end{lemma}

\section{Readability of bipartite chain graphs}\label{sec:cnn-bounds}

In this section, we prove an upper (Section~\ref{subsec:cnn-upperbound}) and
a lower (Section~\ref{subsec:cnn-lowerbound})
bound  on the readability of twin-free bipartite chain graphs, $C_{n,n}$.
Recall that the graph $C_{n,n}$ is $(V_s,V_p,E)$ where $V_s =\{s_1,\ldots,s_n\}$,
$V_p=\{p_1,\ldots,p_n\}$, and $E=\{(s_i,p_j) \mid 1\le i\leq j\le n\}$.

%A {\it bipartite chain graph} is a bipartite graph $G = (V_s,V_p,E)$ such that the vertices in $V_s$ (equivalently: the vertices in $V_p$) can be linearly ordered with respect to inclusion of their neighborhoods, that is, we can write $V_s = \{v_1,\ldots, v_k\}$ such that $N(v_1)\subseteq \ldots\subseteq N(v_k)$. In a bipartite chain graph for every two vertices $u$ and $v$ from the same part we have $N(u)\subseteq N(v)$ or $N(v)\subseteq N(u)$.
%As established by Lemma~\ref{lem:twin-free-reduction}, the readability of a graph is equal to that of its twin-free reduction. Thus, we will lose no generality by presenting the upper bound only for twin-free bipartite chain graphs. Twin-free bipartite chain graphs must have the same number of vertices on either side and there is, up to isomorphism, for each $n\in \N$ exactly one bipartite chain graph with $n$ vertices in each part. In this section, we use $C_{n,n}$ to refer to the bipartite chain graph with $n+n$ vertices. Formally, $C_{n,n}=(V_s,V_p;E)$ where $V_s =\{s_1,\ldots,s_n\}, V_p=\{p_1,\ldots,p_n\}, E=\{(s_i,p_j) \mid 1\le i\leq j\le n\}$.
%In Section~\ref{subsec:cnn-upperbound}, we prove Theorem~\ref{theorem:ub-bip-chain-graphs}: an $\Oh(\sqrt{n})$ upper bound on the readability of $C_{n,n}$.  In Section~\ref{subsec:cnn-lowerbound}, we prove Theorem~\ref{theorem:lb-bip-chain-graphs}: an $\Omega(\log n)$ lower bound on the readability of $C_{n,n}$.

\subsection{Upper bound}\label{subsec:cnn-upperbound}

%\begin{theorem}\label{theorem:ub-cnn}
%	The graph family $C_{n,n}$ has readability $\Oh(\sqrt{n})$, even over an alphabet of size~3.
%\end{theorem}
To prove Theorem~\ref{theorem:ub-bip-chain-graphs}, we construct a labeling $\ell$ of length $\Oh(\sqrt{n})$ for $C_{n,n}$ that satisfies (1) $\ell(s_i) = \ell(p_i)$ for all $i\in [n]$, and (2) $\ell(s_i)$ properly overlaps $\ell(s_j)$ if and only if  $i < j$.
It is easy to see that such an $\ell$ will be a valid overlap labeling of $C_{n,n}$. As the labels on either side of the bipartition are equal, we will just come up with a sequence of $n$ strings to be assigned to one of the sides of $C_{n,n}$ such that the strings satisfy condition (2) above.
%We will call a sequence of strings $(s_1,\dots, s_t)$ {\em forward-matching} if
%(1) $\forall i\in [t],$ $s_i$ does not have a proper overlap with itself, and
%(2) $\forall i,j\in [t],$ $s_i$ overlaps $s_j$ if and only if  $i\leq j$.
\begin{definition}
	A sequence of strings $(s_1,\dots, s_t)$ is {\em forward-matching} if
	\begin{itemize}
		\item $\forall i\in [t],$ string $s_i$ does not have a proper overlap with itself and
		\item $\forall i,j\in [t],$ string $s_i$ overlaps string $s_j$ if and only if  $i\leq j$.
	\end{itemize}
\end{definition}
Given an integer $r\geq 2$, we will show how to construct a forward-matching sequence $S_r$ with $\Theta(r^2)$ strings, each of length at most~$r$, over an alphabet of size $3$.
This will imply an overlap labeling of length $\Oh(\sqrt{n})$ for $C_{n,n}$, proving Theorem~\ref{theorem:ub-bip-chain-graphs}. The following lemma is crucial for this construction.

%\noindent The rest of this section is devoted to constructing a forward-matching sequence of $n$ strings such that each string has length $\Oh(\sqrt{n})$ and is over an alphabet of size $3$. This is enough to prove Theorem~\ref{theorem:ub-bip-chain-graphs}.
%Observe that every string matches itself with a non-proper match.
\begin{lemma}\label{lem:forward-matching-expansion}
	For all integers $t\geq 2$ and all $i\in[t-1]$, if $(s_1,\dots,s_t)$ is forward-matching, so is $(s_1,\dots,s_i,s_is_{i+1},s_{i+1},\dots,s_t)$.
	%If $(A, B, C, D)$ is forward-matching, so is $(A, B, BC, C, D)$.
\end{lemma}
\begin{proof}
	%First, we show that all desired overlaps with the new string $BC$ are formed correctly:
	%\begin{eqnarray*}
	%	&A \text{ overlaps $BC$} & \text{ because $A$ overlaps $B$,}\\
	%	&B \text{ overlaps $BC$} & \text{ by construction,}\\
	%	&BC \text{ overlaps $C$} & \text{ by construction, and}\\
	%	&BC \text{ overlaps $D$} & \text{ because $C$ overlaps $D$.}
	%\end{eqnarray*}
	%\noindent We now show that there are no undesired overlaps.%\snote{Need pictures for the cases below.}
	For the purposes of notation, let $A$ be an arbitrary string from $s_1, \dots, s_{i-1}$ (if it exists),
	let $B=s_i$, $C=s_{i+1}$, and let $D$ be an arbitrary string from $s_{i+2}, \dots, s_t$ (if it exists).
	The reader can easily verify that $A$ and $B$ overlap with the new string $BC$, and $BC$ overlaps with $C$ and $D$, as desired. %all the desired overlaps with the new string $BC$ are formed correctly.
	What remains to show is that there are no undesired overlaps.
	Suppose for the sake of contradiction that $BC$ overlaps $B$, and let $i$ be the length of any such overlap. If $\suf_i(BC)$ only includes characters from $C$, then $C$ overlaps $B$; if it includes characters from $B$ (and the entire $C$) then $B$ has a proper overlap with itself (see Figure~\ref{fig:forward_matching}). In either case, we reach a contradiction.
	So, $BC$ does not overlap $B$.
	By a symmetric argument, $C$ does not overlap $BC$.
	
	\begin{figure}[t]
		%	\begin{center}
		%		\includegraphics[scale = 0.4]{forward-matching.pdf}
		%	\end{center}
		\centering
		\begin{subfigure}{.5\textwidth}
			\centering
			\includegraphics[scale = 0.3]{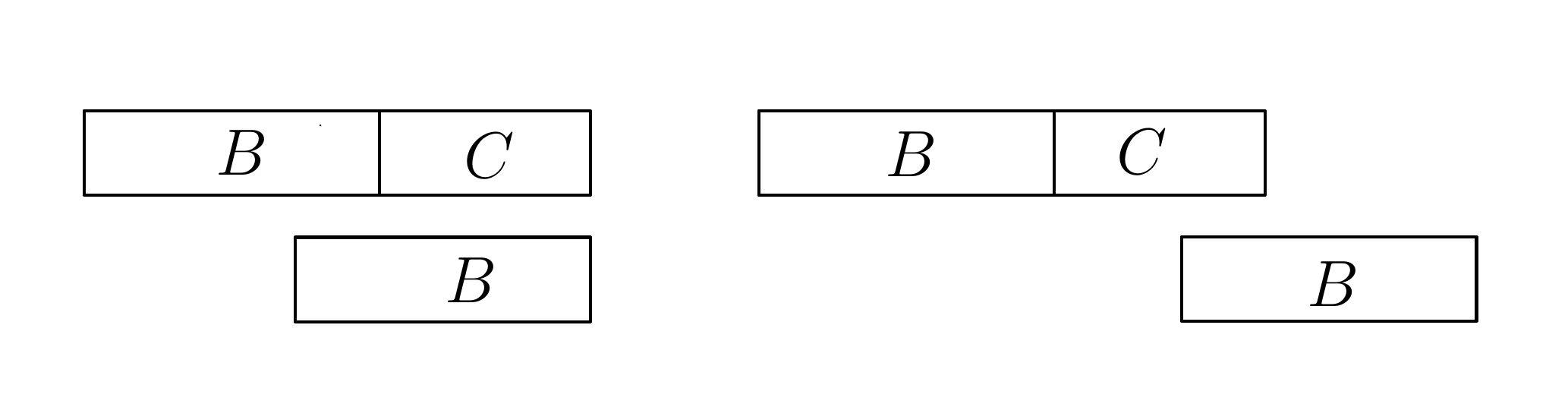}
			\caption{$BC$ does not overlap $B$.}
			\label{fig:forward_matching}
		\end{subfigure}
		\begin{subfigure}{.4\textwidth}
			\centering
			\includegraphics[scale = 0.3]{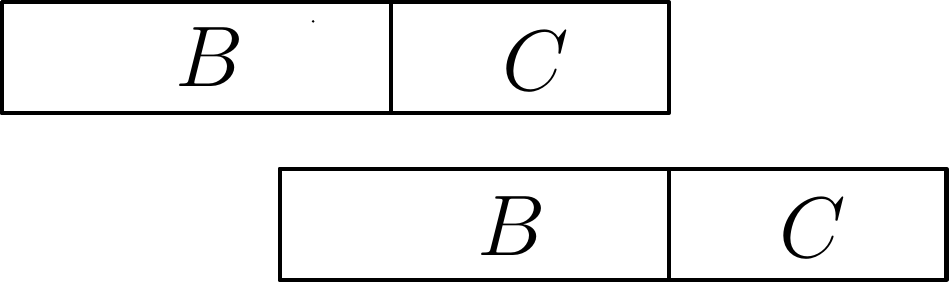}
			\caption{$BC$ has no proper overlap with itself.}
			\label{fig:forward_matching_2}
		\end{subfigure}
		\caption{Overlaps in the proof of Lemma~\ref{lem:forward-matching-expansion}}
	\end{figure}
	
	Next, suppose for the sake of contradiction that $BC$ overlaps $A$, and let $i$ be the length of any such overlap. If $\suf_i(BC)$ only includes characters from $C$, then $C$ overlaps $A$; if it includes characters from $B$ (and the entire $C$) then $B$ overlaps $A$.  In either case, we reach a contradiction. So, $BC$ does not overlap $A$. By a symmetric argument, $D$ does not overlap $BC$.
	
	%\begin{figure}[h!] \begin{center} \includegraphics[scale = 0.4]{forward-matching-2.pdf} \end{center} \caption{$BC$ has no proper overlap with itself.}\label{fig:forward_matching_2} \end{figure}
	
	Finally, suppose for the sake of contradiction that $BC$ has a proper overlap with itself, and let $i$ be the length of any such overlap. Since $C$ does not overlap $BC$, it follows that $\suf_i(BC)$ must include characters from $B$ and the entire $C$. But then $B$ has a proper overlap with $B$, a contradiction (see Figure~\ref{fig:forward_matching_2}). So, $BC$ does not have a proper overlap with itself, completing the proof.
\end{proof}

Now, we show how to construct a forward-matching sequence $S_r$. For the base case, we let $S_2 = (20,0,01)$. It can be easily verified that $S_2$ is forward-matching.
Inductively, let $S_r$ for $r > 2$ denote the sequence obtained from $S_{r-1}$ by applying the operation in Lemma~\ref{lem:forward-matching-expansion} to all indices $i$ such that $s_is_{i+1}$ is of length $r$, that is, add all obtainable strings of length~$r$.
Let $B_r$, for all integers $r\geq 2$, be the sequence of lengths of strings in $S_r$.
We can obtain $B_r$ directly from $B_{r-1}$ by performing the following operation: for each consecutive pair of numbers $x,y$ in $B_{r-1}$, if $x+y=r$ then insert $r$ between $x$ and $y$.
Note that there is a mirror symmetry to the sequences with respect to the middle element, 1. The right sides of the first 6 sequences $B_r,$ starting from the middle element, are as follows:
$$\begin{array}{l | c c c c c c c c c c}
r=2 & 1& 2& \\
r=3 & 1& 3 & 2 \\
r=4 & 1& 4& 3 & 2 \\
r=5 & 1& 5& 4& 3 & 5& 2 \\
r=6 & 1& 6& 5& 4& 3 & 5& 2 \\
r=7 & 1& 7& 6& 5& 4& 7& 3 & 5& 7& 2
\end{array}$$

%For example,
%$$\begin{array}{l | l | l}
%	B_2 = (2, 1, 2)
%	& B_3 = (2, 3, 1, 3, 2)
%	& B_4 = (2,3,4,1,4,3,2).
%\end{array}$$
%Observe that $B_r$ is symmetric, and, in principle, it would suffice to analyze just one of its halves
%(as explained in Section~\ref{sec:techover}).
%For simplicity, we analyze the full sequence here.
%\pnote{I think this is necessary here in order for the reviewer to understand why the sequences in Section 1.2 are different then the ones here}
%Sofya: I think with only ``half''-sequences shown, the connection should be clear.

It turns out that $|B_r|$, and, by extension, $|S_r|$, is closely related to the
totient summatory function \cite{OEIS}, also called the partial sums of Euler's totient function.
This is the function $\Phi(r)=\sum_{k=1}^r\varphi(k),$ where
$\varphi(k)$ is the number of integers in $[k]$ that are coprime to $k$.
%\begin{definition}[Totient summatory function \cite{totient}]\label{def:totient-summatory-function}
%	The {\em totient summatory function} (also called the partial sums of Euler's totient function) is the function $\Phi(r)=\sum_{k=1}^r\varphi(k),$ where
%	$\varphi(k)$ is the number of integers in $[k]$ that are coprime to $k$.
%\end{definition}
The asymptotic behavior of $\Phi(r)$ is well known:
$\Phi(n)=\frac{3n^2}{\pi^2} +\Oh(n \log n)$~\cite[p.~268]{MR568909}.
The following lemma therefore implies $|S_r| = |B_r| = \Theta(r^2)$, completing the proof of~Theorem~\ref{theorem:ub-bip-chain-graphs}.

\begin{lemma}\label{lem:totient}
	For all integers $r\geq 2$, the length of the sequence $B_r$ is $\Phi(r) +1$.
\end{lemma}
\begin{proof}
	For the base case, observe that $|B_2| = 3 = \Phi(2) + 1$. %The length of the sequence $B_2$ is 3, which is $\Phi(2)+1$.
	In general, consider the case of $r\geq 3$.
	%\snote{You can not define pair:= neighbors. There is a mismatch in the number of objects. This is what I found here (not sure why it was changed): Define an ordered pair of integers as {\em neighbors} if they occur consecutively in $B_r$.}
	\begin{definition}
		Two elements of $B_r$ are called {\em neighbors in $B_r$} if they appear in two consecutive positions in $B_r$.
	\end{definition}
	We will show that any two neighbors are coprime (Claim~\ref{claim:neighbors-are-coprime}) and
	any pair $(i,j)$ of coprime positive integers that sum up to $r$ appears exactly once %\snote{added: exactly once}
	as a pair of ordered neighbors in $B_r$ (Claim~\ref{claim:exactly-once}).
	Together, these claims show that the neighbor pairs in $B_{r-1}$ that sum up to $r$ are exactly the pairs of coprime positive integers that sum up to $r$.
	\begin{fact}\label{fact:rel-prime}
		If $i$ and $j$ are coprime then each of them is coprime with $i+j$ and with $i-j.$
	\end{fact}
	By this fact, there is a bijection between pairs $(i,j)$ of coprime positive integers that sum up to $r$ and integers $i\in [r]$ that are coprime to $r$.
	Hence, the number of neighbor pairs in $B_{r-1}$ that sum up to $r$ is $\varphi(r)$.
	Therefore, $B_r$ contains $\varphi(r)$ occurrences of $r$.
	By induction, it follows that $|B_r|= |B_{r-1}| + \varphi(r) = \Phi(r-1)+1 + \varphi(r)= \Phi(r)+1$, proving the Lemma.
\end{proof}
We now prove the necessary claims.
	\begin{claim}\label{claim:neighbors-are-coprime}
		For all $r\geq 2$, if two numbers are neighbors in $B_r$, they are coprime.
	\end{claim}
	\begin{proof}
		We prove the claim by induction. For the base case of $r=2$, the claim follows from the fact that 1 and 2 are coprime.
		For the general case of $r \geq 3$, recall that $B_r$ was obtained from $B_{r-1}$ by inserting an element $r$ between all neighbors $i$ and $j$ in $B_{r-1}$ that summed to $r$.
		By the induction hypothesis, $gcd(i,j) = 1$, and, hence, by Fact~\ref{fact:rel-prime}, $gcd(i,r) = gcd(i,i+j) =1$ and $gcd(r,j) = gcd(i+j, j) = 1$.
		Therefore, any two neighbors in $B_r$ must be coprime.
	%Now suppose the claim holds for $B_{r-1}$ for some $r\geq 3$. We will show that it also holds for $B_r$. Recall that $B_r$ was obtained from $B_{r-1}$ by inserting elements $r$ between neighbors in $B_{r-1}$ that summed to $r$. Consider two neighbors $i,j$ in $B_{r-1}$ such that $i+j=r$. Since $gcd(i,j)=1$, by the inductive hypothesis, Fact~\ref{fact:rel-prime} gives that $gcd(i,i+j)=gcd(i+j,j)=1$.
	\end{proof}
	
	\begin{claim}\label{claim:exactly-once}
		For all $r\geq 3$, every ordered pair $(i,j)$ of coprime positive integers that sum to $r$
		occurs exactly once as neighbors in $B_{r-1}$.
	\end{claim}
	\begin{proof}
		We prove the claim by strong induction. The reader can verify the base case (when $r=3$). %We have $B_2=(2,1,2).$ The sequence $B_2$ has two ordered pairs of neighbors that sum up to 3: namely, $(2,1)$ and $(1,2)$. Each of them appears exactly once. There exist no other pairs $(i,j)$ of positive integers with $i+j=3$. So, the claim holds for $r=3$.
		For the inductive step, suppose the claim holds for all $k\leq r-1$ for some $r\geq 4$.
		%We will prove it for $r=k$.
		Consider an ordered pair $(i,j)$ of coprime positive integers that sum to $r$.
		Assume that $i > j$; we know that $i \neq j$, and the case of $i < j$ is symmetric.
		Since $r\geq 4$, we have that $i\geq 3$.
		%Recall that $i$ is not present in the sequences $B_r$ for $r<i$.
		In the recursive construction of the sequences $\{B_k\}$, the elements $i$ are added to the sequence $B_i$ when $B_i$ is created from $B_{i-1}$.
		Since $j<i$, all the elements $j$ are already present in $B_{i-1}$.
		%Consequently, for each possible neighbor pair $(i,j)$ in $B_{k-1}$, the element $i$ is obtained by adding adjacent elements $i-j$ and $j$ in $B_{i-1}$.
		By Fact~\ref{fact:rel-prime}, since $gcd(i,j)=1$, we get that $gcd(i-j,j)=1$.
		By the inductive hypothesis, pair $(i-j,j)$ appears exactly once as an ordered pair of neighbors in $B_{i-1}$.
		%By the inductive hypothesis, the neighbor pair $(i-j,j)$ appears exactly once in $B_{i-1}$.
		Consequently, $(i,j)$ must appear exactly once as an ordered pair of neighbors in $B_i$.
		No new elements $i,j$ are added to the sequence in later stages, when $k>i$.
		Also, no new elements are inserted between $i$ and $j$ when $i+1\le k\le i+j-1 = r - 1$.
		Therefore, the ordered neighbor pair $(i,j)$ appears exactly once in $B_{r-1}$.
		%By strong induction, Claim~\ref{claim:exactly-once} holds.
	\end{proof}

\subsection{Lower bound}\label{subsec:cnn-lowerbound}

In this section, we prove Theorem~\ref{theorem:lb-bip-chain-graphs}, namely
that the readability of $C_{n,n}$ is $\Omega(\log n)$.
First, we will need the notion of a HUB decomposition from~\cite{ChikhiMMR16}.
Given $G = (V_s, V_p, E)$ and a function $w: E \to [k]$,
we define $G_i$, for $i\in[k]$, as the graph with the same vertex set as $G$ and edges given by
$E(G_i) = \setof{e \in E}{w(e) = i}$.
Observe that the edge sets of $G_1,\dots, G_k$ form a partition of $E$.
We say that $w$ is a {\em hierarchical-union-of-bicliques decomposition}, abbreviated as {\em HUB decomposition},
if the following conditions hold:
i) for all $i\in[k]$, $G_i$ is a disjoint union of bicliques, and
ii) if two distinct vertices $u$ and $v$ are non-isolated twins in $G_i$ for some
$i\in\{2,\dots,k\}$ then, for all $j\in [i-1]$,
$u$ and $v$ are (possibly isolated) twins  in $G_j$.
The parameter $k$ is called the size of the decomposition $w$.
Now, consider a HUB decomposition of $C_{n,n}$ of size $h$.

\begin{lemma}\label{lemma:max-degee-in-hub-decomposition}
	For each $i\in\{0,\dots, h-1\}$, graph $G_{h-i}$  has maximum degree at most~$2^i$.
\end{lemma}
\begin{proof}
	We prove the lemma by strong induction on $i$.
	The base case is when $i=0$.
	Observe that if $G_h$ has non-isolated twins, then those must be twins in $G_j$ for each $j \in [h]$,
	and, as a result, in $C_{n,n}$.
	Since $C_{n,n}$ has no twins, $G_h$ has no non-isolated twins.
	By the first property of the HUB decomposition, $G_h$ must have maximum degree at most~1.

	%This holds because $C_{n,n}$ has no twins. Recall that in a HUB decomposition, if two vertices are non-isolated twins in $G_j$, they must be twins in $G_k$ for all $k\in [j]$. So, if two vertices are twins in $G_h$, they must be twins in the whole graph $C_{n,n}$. Therefore, $G_h$ has maximum degree at most~1.
		
	For general $i$,
	let $F_i$ denote the graph $(V_s,V_p,\bigcup_{j\in\{0,\dots,i-1\}} E(G_{h-j}))$.
	By the inductive hypothesis, $F_i$ has maximum degree at most~$\sum_j 2^j=2^i-1$.
	Consider a group of vertices $S$ in the same part of $C_{n,n}$ that have the same degree in the graph $C_{n,n} - E(F_i)$.
	Since no two vertices in the same part of $C_{n,n}$ have the same degree, no two vertices in $S$ have the same degree in $F_i$. Combining this with the fact that the degree of any vertex in $F_i$ is at most~$2^i-1$, we infer that $|S| \leq 2^i$.

	By the second property of the HUB decomposition, if two vertices are non-isolated twins in $G_{h-i}$,
	they are twins in $C_{n,n} - E(F_i)$.
	Consequently, each group of twins in $G_{h-i}$ has size at most~$2^i$.
	By the first property of the HUB decomposition, $G_{h-i}$ is a disjoint union of bicliques.
	It follows that each of these bicliques is a subgraph of the complete bipartite graph $K_{2^i,2^i},$
	implying the required bound on the maximum degree.
	%For the inductive step, suppose that for some $i$ and all $j\in\{0,\dots, i-1\}$, the maximum degree of $G_{h-j}$ is at most~$2^j$. We need to prove that the maximum degree of $G_{h-i}$ is at most~$2^i$. Let $F_j$ denote the graph $(V_s,V_p,\bigcup_{j\in\{0,\dots,i-1\}} E(G_{h-j}))$.
	%By the inductive hypothesis, graph $F_j$ has maximum degree at most~$\sum_{j\in\{0,\dots,i-1\}} 2^j=2^i-1$.
	%Consider a group of vertices $S$ in the same part of $C_{n,n}$ that have the same degree in the graph $C_{n,n} - E(F_j)$. Since no two vertices in the same part of $C_{n,n}$ have the same degree, no two vertices in $S$ have the same degree in $F_j$. Combining this with the fact that the degree of any vertex in $F_j$ is at most~$2^i-1$, we infer that in $C_{n,n} - E(F_j)$, any group of vertices in the same part of the bipartition of the same degree has size at most~$2^i$.
	%Consequently, each group of twins in $G_{h-i}$ has size at most~$2^i$. Recall that $G_{h-i}$ is a disjoint union of bicliques. It follows that each of these bicliques is a subgraph of the complete bipartite graph $K_{2^i,2^i},$ implying the required bound on the maximum degree.
\end{proof}
\begin{proof}[Proof of Theorem~\ref{theorem:lb-bip-chain-graphs}]
	By Lemma~\ref{lemma:max-degee-in-hub-decomposition}, graph $G_{h-i}$ has at most~$2^i n$ edges.
	Since the edge sets of $G_1,\dots, G_h$ form a partition of the edge set of $C_{n,n}$,
	%the edge set of $C_{n,n}$ is the disjoint union $E(C_{n,n}) = \bigcup_{i\in\{0,\dots,h-1\}} E(G_{h-i})$,
	the number of edges in $C_{n,n}$ is
	$\frac{n(n+1)}2\leq \sum^{h-1}_{i=0} 2^i n=n(2^h-1).$
	%$\frac{n(n+1)}2\leq \sum_{i\in\{0,\dots,h-1\}} 2^i n=n(2^h-1).$
	We get that $h\geq \log_2(n+3)-1$. It was shown in~\cite{ChikhiMMR16} that the readability of every bipartite graph $G$
  is bounded from below by the minimum size of a HUB decomposition of $G$.
	This completes the proof.
	%$\geq {\it hub}(G)$ for all graphs $G$. We use this fact to show the following.
	%By $r(G) \geq {\it hub}(G)$\cite{ChikhiMMR16}, it suffices to show that ${\it hub}(C_{n,n})=\Omega(\log n)$.
	%Let $C_{n,n} = (V_s,V_p,E)$, let $h = {\it hub}(C_{n,n})$, let $w:E(C_{n,n})\to [h]$ be
	%a HUB decomposition of $C_{n,n}$ of minimum size, and
	%consider the graphs $G_i = G_i^w$, for $i\in [h]$. We will show that $h\geq \log_2(n+3)-1$.
\end{proof}

\section{A characterization of graphs with readability at most~2}\label{sec:readability-2-characterization}

In this section, we characterize bipartite graphs with readability at most~$2$ by proving Theorem~\ref{thm:readability-2-characterization}. Due to Lemma~\ref{lem:induced-subgraph}, it is enough to obtain such a characterization for connected twin-free bipartite graphs.
%For twin-free graphs, we show that having readability at most~$2$ is equivalent to the existence of a particular matching (with properties specified in Definition~\ref{def:feasible-matching}).
We use this characterization in Section~\ref{sec:readability-2} to develop a polynomial time algorithm for recognizing graphs of readability at most~$2$ and also in Section~\ref{sec:grids} to prove a lower bound on the readability of general grids.
Recall that a {\it domino} is the graph obtained from $C_6$ by adding an edge between two vertices at distance~$3$. 
%For twin-free graphs, we show using Lemma~\ref{lem:graph} that having readability at most~$2$ is equivalent to the existence of a particular matching (with properties specified in Definition~\ref{def:feasible-matching}).
We first define the notion of a feasible matching, which is implicitly used
in the statement of Theorem~\ref{thm:readability-2-characterization}.
%This is enough to prove Theorem~\ref{thm:readability-2-characterization}.

	\begin{definition}\label{def:feasible-matching}
		A matching $M$ in a bipartite graph $G$ is {\em feasible} if the following conditions are satisfied:
		\begin{enumerate}
			\item The graph $G' = G-M$ is a disjoint union of bicliques (equivalently: $P_4$-free).
			\item For $U\subseteq V(G)$, if $G[U]$ is a $C_6$, then $G'[U]$ is the disjoint union of three edges.
%\mnote{Just a suggestion (I would not insist on it): technically speaking, an edge is not a graph. It would be more correct to write  $G'[U]\cong 3K_2$ (of course, after having introduced notation $\cong$ for the graph isomorphism relation and after having defined the graph $3K_2$ as the disjoint union of three $K_2$s). Similarly, in condition 3, we could write $G'[U]\cong C_4+K_2$. This comment, if implemented, would also affect the formulation of Theorem~\ref{thm:readability-2-characterization}.}
			\item For $U\subseteq V(G)$, if $G[U]$ is a domino, then $G'[U]$ is the disjoint union of a $C_4$ and an edge.
		\end{enumerate}
	\end{definition}

We prove Theorem~\ref{thm:readability-2-characterization} by showing that a bipartite graph $G$ has readability at most~$2$ iff $G$ has a feasible matching.

\begin{proof}[Proof of Theorem~\ref{thm:readability-2-characterization}]
We show that $r(G)\le 2$ if and only if $G$ has a feasible matching.

\medskip
\noindent{\it Necessity.} Suppose that $G = (V_s, V_p, E)$ is a twin-free bipartite graph of readability at most~$2$. Let $\ell$ be an overlap labeling of $G$ of length at most~$2$. Since $\ell$ is an overlap labeling of $G$, we can partition the edge set of $G$ into two sets, $E_1$ and $E_2$, by setting
		$E_1 = \{(u,v)\in E\mid \ov_\ell(u,v) = 1\}$ and $E_2 = E\setminus E_1$. Then
		for all $(u,v)\in E_2$, we have $\ov_\ell(u,v) = 2$, that is, $\ell(u) = \ell(v)$.
		Note that due to the definition of the overlap function,
		for every edge $(u,v)\in E_2$, the labels of $u$ and $v$ must not have an overlap of length one.
		
		We claim that $E_2$ is a feasible matching.
		If $E_2$ is not a matching, we can assume by symmetry
		that there exists a vertex $u\in V_s$ and a pair of distinct vertices $v,w$ in $V_p$ such that
		$\{(u,v),(u,w)\}\subseteq E_2$. But then $\ell(v) = \ell(u) = \ell(w)$, which implies that $v$ and $w$ are twins in $G$,
		a contradiction. Thus, $E_2$ is a matching.
		
		Let $G'$ denote the graph $G - E_2$. Next, we show that $G'$ is $P_4$-free. If $(u,v,x,y)$ forms an induced $P_4$ in $(V,E_1)$ (with edge set $\{(u,v), (x,v),(x,y)\}$), then
		$\suf_1(\ell(u)) = \pre_1(\ell(v)) = \suf_1(\ell(x)) = \pre_1(\ell(y))$,
		implying that $(u,y)\in E_1$, a contradiction.
		Therefore, $G'$ is $P_4$-free.
		
		Now let us verify the remaining two properties in the definition of a feasible matching.
		Let $U$ be a subset of vertices in $G$. If $G[U]$ is isomorphic to $C_6$, we would like to show that $G'[U]$ is a union of three disjoint edges. Suppose for the sake of contradiction that it is not.
		Consider an edge labeling of $G[U]$ as in Figure~\ref{fig:C6-domino-fork}. Since $E_2$ is a matching, the only other way for $G'$ to be $P_4$-free, i.e., if it was not a union of three disjoint edges, is for $E_2$ to contain two diametrically opposite	edges of $G[U]$, say $e_1$ and $e_4$. Let $e_i = (x_i,x_{i+1})$ for all $i \in [6]$ (addition modulo $6$). Let, without loss of generality, $x_1 \in V_s$.%\mnote{I slightly simplified this part of the proof.}
		Then $x_2 \in V_p$. Since $e_1 \in E_2$ by our assumption, we have $\ell(x_1) = \ell(x_2)$, say $\ell(x_1) = \ell(x_2)= ab$. We have $\suf_1(\ell(x_5)) = \pre_1(\ell(x_6)) = \suf_1(\ell(x_1)) = b$ and $\pre_1(\ell(x_4)) = \suf_1(\ell(x_3)) = \pre_1(\ell(x_2)) = a$. Since $e_4 \in E_2$, we get $\ell(x_4) = \ell(x_5) = ab$. Therefore $\ell(x_1) = \ell(x_4)$, which is a contradiction, since $(x_1,x_4) \notin E$ and $\ell$ is an overlap labeling of $G$.
		
\begin{figure}[t]
	\centering
	\includegraphics{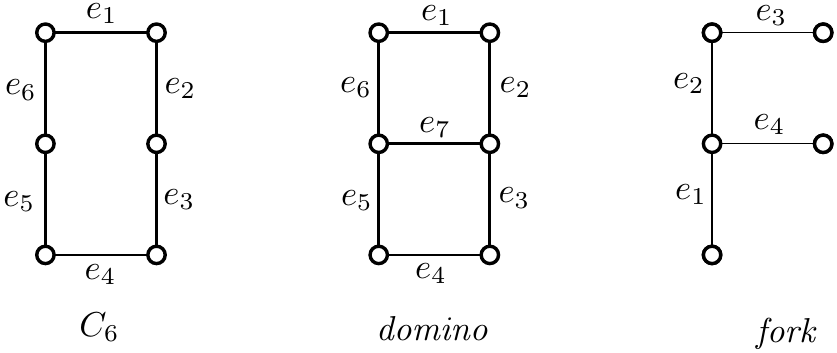}
	\caption{The $C_6$, the domino and the fork}
	\label{fig:C6-domino-fork}
\end{figure}
		
		Finally, suppose that $G[U]$ is isomorphic to the domino,
		and assume an edge labeling as in Figure~\ref{fig:C6-domino-fork}.
		Since $G'$ is $P_4$-free, $G'[U]$ is also $P_4$-free and hence
		$G'[U]$ can only be isomorphic to either (1) a disjoint union of a $C_4$ and an edge (which is what we want to show), or (2) a disjoint union of two $P_3$'s. Suppose we are in case (2). Then we have $e_1,e_4,e_7 \in E_2$.
		Let $e_i = (x_i,x_{i+1})$ for all $i \in \{1,\ldots, 6\}$ (addition modulo $6$). We may assume without loss of generality that $x_1\in V_s$.
		Since $e_1,e_4 \in E_2$ and $e_2,e_3,e_5,e_6 \in E_1$, we can follow the same reasoning as above, and conclude that the labels of $x_1$ and $x_4$ are equal, which is a contradiction, since $(x_1,x_4) \notin E$ and $\ell$ is an overlap labeling of $G$.
This establishes the necessity of the condition.
		
		%Similarly as above, assuming that
		%$\ell(x_1) = \ell(x_2) = ab$ and considering the edges
		%$e_2$, $e_3$,$e_6$, and $e_7$, the same reasoning as in the case of $F$ being isomorphic to $C_6$, we infer
		%that $\ell(x_5) = \ell(x_4) = ab$, leading to a contradiction with $(x_5,x_2)\not\in E$.
		%This shows that $E_2\cap E(F)\in \{\{e_2,e_6\},\{e_3,e_5\}\}$.
		
		\medskip
\noindent{\it Sufficiency.} Suppose now that $G = (V_s, V_p, E)$ is a twin-free bipartite graph with a feasible matching $M$.
		We will show that $G$ has readability at most~$2$ by constructing an overlap  labeling of $G$ of length at most~$2$.
		Since $M$ is a feasible matching, the graph $G' = G-M$ is $P_4$-free, that is, a disjoint union of bicliques.
		Let $\{A_1,B_1\},\ldots, \{A_k,B_k\}$ be the bipartitions of
		the vertex sets of the connected components (bicliques) $G_1,\ldots, G_k$ of $G'$ (so that $A_i = V(G_i)\cap V_s$ for all $i$; some of the $A_i$'s or $B_i$'s
		may be empty). %\nvnote{There is no loss of generality in not considering 'empty' bicliques, right? Or, am I missing something?}\mnote{If could be that there are some isolated vertices on either side of $G'$. Each isolated vertex corresponds to a biclique with one of the $A_i$ and $B_i$ empty (and the other one a singleton)}.
		Then $\cup_{i = 1}^k V(G_i) = V$.
		Assign a partial labeling over $\Sigma = \{1,\ldots, k\}$
		to vertices of $G$ by setting $\ell(v) = i$ if and only if $v\in V(G_i)$.
		For each edge $(u,v)\in M$, extend the labels of $u\in V_s$ and $v\in V_p$ as
		follows.
		Let $u\in A_i$ and $v\in B_j$. Then $i\neq j$ because edges of bicliques in $G-M$ cannot be in $M$.
		Replace $\ell(u) = i$ with $\ell(u) = ji$,
		and $\ell(v) = j$ with $\ell(v) = ji$.
		Since $M$ is a matching, every vertex will have a label of length $1$ or $2$ at the end of this procedure.
		Extend the labels of length $1$ by unique new characters to make them of length $2$.
		By construction, the overlaps of the obtained labeling create all
		edges of $E(G')\cup M = E(G)$.
		
		Let us verify that no new edges were created by $\ell$.
		Suppose that $u,v$ is a pair of vertices with
		with $u\in V_s$ and $v\in V_p$ and $\ov_\ell(u,v)>0$.
		If $\ell(u)$ and $\ell(v)$ have an overlap of length $1$, then $(u,v)\in E(G')$ by construction.
		Suppose that $\ell(u)$ and $\ell(v)$ do not have an overlap of length $1$ but have an overlap of length $2$. Then
		$\ell(u)= \ell(v) = ij$ for two distinct $i,j\in \Sigma$.
		By construction, vertex $u$ is adjacent to a unique vertex $w$ via a matching edge in $M$,
		moreover $u\in A_j$ and $w\in B_i$. If $w = v$, then  the edge $(u,v)$ is in $M$ and hence in $G$.
		So we may assume that $w\neq v$. Similarly, vertex $v$ is adjacent to a unique vertex $z$ in $M$, and $z\in
		A_j$ and $v\in B_i$. If $u = z$, then again the edge $(u,v)$ is in $M$ and hence in $G$.
		So we may assume that $u\neq z$. Since $|A_j|\ge 2$, there exists a vertex $s\in B_j$.
		Similarly, since $|B_i|\ge 2$, there exists a vertex $t\in A_i$. Notice that $(u,v)\not\in M$ since $u$ is of degree $1$ in $M$, and
		$(u,v)\not\in E(G')$ since $u$ and $v$ belong to distinct connected components of
		$G'$. Therefore, $(u,v)\not\in E(G)$, and, similarly, $(z,w)\not\in E(G)$. But now, the subset $\{s,t,u,v,w,z\}$ induces a subgraph of $G$ isomorphic to either a $C_6$ (if $(t,s)\not\in E(G)$) or a domino (otherwise). In either case, one of the
		conditions for the $C_6$ and for the domino in Definition~\ref{def:feasible-matching} is violated, contrary to the fact that $M$ is
		a feasible matching.
		
		This shows that $\ell$ is an overlap labeling of $G$ and implies that the readability of $G$
		is at most~$2$.
	\end{proof}

\begin{corollary}\label{cor:path-cycle}
	Every bipartite graph $G$ of maximum degree at most~$2$ has readability at most~$2$.
\end{corollary}

\begin{proof}
If $G$ is a connected twin-free bipartite graph of maximum degree at most~$2$, then $G$ is a path or an (even) cycle.
In this case, the edge set of $G$ can be decomposed into two matchings $M_1$ and $M_2$. Both $M_1$ and $M_2$ are feasible matchings. Thus, by
Theorem~\ref{thm:readability-2-characterization},
$G$ has readability at most~$2$.
\end{proof}

\section{An efficient algorithm for readability $\boldsymbol{2}$}\label{sec:readability-2}

In this section, we prove Theorem~\ref{thm:readability-2} by developing a polynomial time algorithm for the following problem.

\begin{center}
	\fbox{\parbox{0.89\linewidth}{\noindent
			{\sc Readability $2$}
			\\[.8ex]
			\begin{tabular*}{.95\textwidth}{rl}
				{\em Instance:} & A bipartite graph $G = (V_s,V_p,E)$.\\
				{\em Question:} & Is $r(G)\le 2$?
			\end{tabular*}
	}}
\end{center}
\medskip

%Our approach can be summarized as follows.
%As we can determine the maximum degree of a graph in linear time, we restrict our attention to graphs of maximum degree at least $3$. %In this case, we show that the problem can be reduced to 2SAT as follows.

First, we use Lemma~\ref{lem:induced-subgraph} and Corollary~\ref{cor:path-cycle}
to reduce the problem to connected twin-free bipartite graphs of maximum degree at least 3.
We then apply
Theorem~\ref{thm:readability-2-characterization}
and reduce the problem to checking for the existence of a feasible matching (Definition~\ref{def:feasible-matching}).
Finally, we show how  to reduce this problem to the 2SAT problem (Lemma~\ref{lem:2SAT}), which is well known to be solvable in linear time~(see, e.g.,~\cite{MR526451}).

\begin{proof}[Proof of Theorem~\ref{thm:readability-2}.]
Given a bipartite graph $G$, we first reduce the problem to its connected components.
That is, if $G$ is not connected, then,
by Lemma~\ref{lem:induced-subgraph}{(b)},
$r(G)\le 2$ if and only if all components $G'$ of $G$ satisfy $r(G')\le 2$.
Second, assuming $G$ is connected, we compute the twin-free reduction $G'$ of $G$,
which, by Lemma~\ref{lem:induced-subgraph}{(d)}, does not change the readability.
We test whether $G'$ is of maximum degree at most~$2$.
If this is the case, then,
by Corollary~\ref{cor:path-cycle},
we assert that $G$ has readability at most~$2$.

Consider a connected twin-free bipartite graph $G = (V,E)$ of maximum degree at least~$3$. Let $E'$ denote the set of all edges $e=(u,v)$ in $G$ such that either (1) $\{u,v\}\cup N(u) \cup N(v)$ has a vertex of degree at least $3$, or (2) $e$ is contained in some induced $C_6$. The definition of $E'$ and the fact that $G$ is connected and of maximum degree at least $3$ imply
that if an induced subgraph $H$ of $G$ is isomorphic to a $C_4$, a fork, a $C_6$, or a domino (see Figure~\ref{fig:C6-domino-fork}), then $E(H)\subseteq E'$.
%, and
%let $V'$ denote the (nonempty) set of vertices of $G$ of degree at least $3$ in $G$.
%The variables of the $2$-SAT instance we will create will be indexed over a subset $E'$ of edges defined below.
%Given two sets of vertices $X$ and $Y$, we denote by $d(X,Y)$ the distance between $X$ and $Y$, that is,
%the length of a shortest path from a vertex in $X$ to a vertex in $Y$.
%Let $E'$ denote the set of all edges $e$ of $G$ such that either $d(e,V')\le 1$ or $e$ is contained in some induced $6$-cycle.
%In other words, an edge $e$ is in $E\setminus E'$ if and only if $d(e,V')\ge 2$ and $e$ does not belong to any induced $6$-cycle.

Let $X = \{x_e\mid e\in E'\}$ be a set of variables. We now define a 2SAT formula $\varphi$ over $X$ such that $G$ has a feasible matching (and hence, readability at most~$2$) if and only if $\varphi$ is satisfiable. The formula $\varphi$ contains the following five types of clauses.

%An assignment of values true or false to the variables in $X$ corresponds to a subset of edges of $E'$, namely the edges $e$ such that $x_e$ takes value true. We will define a set of clauses over $X$ such that every assignment that satisfies all the clauses will correspond to a matching in $G$ that can be extended (in polynomial time) to a feasible matching of $G$. Conversely, for every feasible matching $M$, an assignment that sets $x_e$ to true iff $e\in M$ will be a satisfying assignment.
%\nithin{} We will define a set of clauses over $X$ such that every satisfying assignment will result in a matching in $G$ that will be almost feasible, in the sense that we will be able to extend it (in polynomial time) to a feasible matching of $G$. And conversely, for every feasible matching of $M$, the assignment of values true or false to the variables in $X$ such that
%$x_e$ is set to true if and only if $e$ is a matching edge will form a satisfying assignment.

%Let $F$ denote an induced subgraph of $G$.
\begin{enumerate}
	\item For each pair $\{e,f\}\subseteq E'$ of distinct edges that share an endpoint, add the clause $\overline{x_e} \vee \overline{x_f}$ to $\varphi$.
	\item For each induced subgraph $H$ of $G$ isomorphic to $C_4$ and each matching $\{e,f\}$ in $H$, add the clauses $\overline{x_e} \vee x_f$ and $\overline{x_f} \vee x_e$ (equivalent to $x_e\leftrightarrow x_f$) to $\varphi$.
	\item For each induced subgraph $H$ of $G$ isomorphic to $C_6$, with edges labeled as in Figure~\ref{fig:C6-domino-fork},
	add the clause $x_{e_1}\vee x_{e_2}$,
	the clauses corresponding to $x_{e_1}\leftrightarrow x_{e_3}$ and $x_{e_3}\leftrightarrow x_{e_5}$, and
	the clauses corresponding to $x_{e_2}\leftrightarrow x_{e_4}$ and $x_{e_4}\leftrightarrow x_{e_6}$ to $\varphi$. %SK: corrected this because $x_{e_2}\leftrightarrow x_{e_4}\leftrightarrow x_{e_6}$ does not have the same meaning, i.e. it is true whenever an odd number of the variables is set to true
	\item For each induced subgraph $H$ of $G$ isomorphic to the domino, with edges labeled as in Figure~\ref{fig:C6-domino-fork},
	add the clauses $x_{e_2}\vee x_{e_3}$ and $x_{e_5}\vee x_{e_6}$ to $\varphi$.
	\item For each induced subgraph $H$ of $G$ isomorphic to the fork, with edges labeled as in Figure~\ref{fig:C6-domino-fork},
	add the clause $x_{e_2}\vee x_{e_3}$ to $\varphi$.
\end{enumerate}

The following lemma shows that if $\varphi$ is satisfiable, then $r(G)\le 2$, otherwise, $r(G)>2$.
\begin{lemma}\label{lem:2SAT}
Graph $G$ has a feasible matching if and only if formula $\varphi$ is satisfiable.
\end{lemma}
\begin{proof}
Suppose first that $G$ has a feasible matching, say $M$.
Let $a$ be an assignment of Boolean values to the variables in $X$ such that
for every $e\in E'$, variable $x_e$ is true if and only if $e\in M$.
We will prove that $a$ is a satisfying assignment for $\varphi$.
It is easy to see that clauses of type (1) in $\varphi$ are satisfied as $M$ is a matching.

Consider a pair of clauses $\overline{x_e} \vee x_f$ and $\overline{x_f} \vee x_e$ of type (2) in $\varphi$.
These correspond to an induced subgraph $H$ of $G$ isomorphic to a $C_4$ and a matching $\{e,f\}$ in $H$.
Since $M$ is a feasible matching, the graph $G - M$ is $P_4$-free, and so we have $e\in M$ if and only if $f\in M$. Hence $a$ satisfies both the clauses.

Clauses in $\varphi$ of type (3) deal with induced $6$-cycles and those of type (4) deal with induced dominos. Both types of clauses are satisfied by $a$ due to
the fact that $M$, which is a feasible matching, satisfies conditions 2 and 3 in Definition~\ref{def:feasible-matching}.

Finally, clauses in $\varphi$ of type (5) are satisfied only if for each induced subgraph $H$ of $G$ isomorphic to the fork
(with edges labeled as in Figure~\ref{fig:C6-domino-fork}), we have
$\{e_2,e_3\}\cap M\neq \emptyset$.
Suppose for the sake of contradiction that there exists an induced fork $H$ for which this is not the case.
%Suppose that this is not the case, and let $H$ be an induced fork in $G$ with
%$\{e_2,e_3\}\cap M= \emptyset$.
Since $G-M$ is $P_4$-free, so is $H-M$ and hence
$e_1$ and $e_4$ are both in $M$, which is a contradiction.
This shows that formula $\varphi$ is satisfiable.

\medskip
For the converse direction, suppose that formula $\varphi$ is satisfiable and let $a$ be a satisfying assignment.
Let $M'$ be the set of edges $e\in E'$ such that $x_e$ is set to true in $a$.
Extend $M'$ greedily to a set of edges $M$ by setting $M= M'$ and then iteratively adding
the middle edge of any induced subgraph $H$ of $G$ isomorphic to $P_4$ that contains no edge of $M$.
We claim that the so obtained set $M$ is a feasible matching of $G$.
This will be easy to show once we prove the following claim.

\begin{claim}
	$M$ is a matching in $G$ with $M\cap E' = M'\cap E'$.
\end{claim}

\begin{proof}%[Proof of Claim.]
	The claim is true if $M = M'$, since $M'$ is a matching by virtue of type (1) clauses. Henceforth, assume that $M\neq M'$. We will first show that $M\cap E' = M'\cap E'$. For this, it is enough to prove that $e \notin E'$ for each $e \in M\setminus M'$.
	
	Consider an edge $e \in M\setminus M'$. By our construction of $M$, the edge $e$ is the middle edge of an induced subgraph $H$ of $G$ isomorphic to $P_4$ such that $H$ contains no other edge of $M$. In particular, $H$ has no edge of $M'$. Let $u$ and $v$ be the endpoints of $e$ and let $x$ and $y$ be the remaining two vertices in $H$ such that $(x,u)$ and $(v,y)$ are the other two edges in $H$. Assume for the sake of contradiction that $e \in E'$. Then, either (a) $\{u,v\}\cup N(u) \cup N(v)$ has a vertex of degree at least $3$, or (b) $e$ is contained in some induced $C_6$.
	
	Suppose first that (b) holds. Then, by virtue of the type (3) clauses, either $(u,v) \in M'$ or both $(x,u)$ and $(y,v)$ are in $M'$. Both cases contradict our premise that $H$ contains no edge of $M'$.
	
	Suppose now that (a) holds. Assume that the degree of $u$ is at least $3$. Let $w$ be a neighbor of $u$ such that $w\neq x$. We will show that the set $\{x,u,w,v,y\}$ induces a fork in $G$. Since $G$ is a bipartite graph, it has no $C_3$'s and hence $(w,x),(w,v) \notin E$. If $(w,y)\in E$, then the set $\{w,u,v,y\}$ induces a $C_4$.
	%As $G$ is twin-free, either $w$ or $v$ has to have  degree at least $3$.
	Since $u$ is of degree at least~$3$, we have $\{(w,u),(u,v),(v,y),(y,w)\} \subseteq E'$ and hence, by virtue of clauses of type (2), either $(u,v)$ and $(w,y)$ are in $M'$, or both $(u,w)$ and $(v,y)$ are in $M'$. Both of these contradict our premise that $H$ contains no edge of $M'$. Therefore, the set $\{x,u,w,v,y\}$ induces a fork in $G$, and by virtue of its associated type (5) clause, either $(u,v)$ or $(v,y)$ is in $M'$. This contradicts our assumption that $H$ does not have any edge in $M'$. Thus, the degree of $u$ is~$2$. By a symmetric argument, the degree of $v$ is~$2$. Thus, the only way for (a) to hold is for either $x \in N(u)$ or $y \in N(v)$ to have degree at least $3$. By symmetry, we may assume that $x$ has degree at least $3$. Let $s,t \in N(x)\setminus\{u\}$. Since $v$ is of degree~$2$, it is non-adjacent to both $s$ and $t$, hence the set $\{s,t,x,u,v\}$ induces a fork in $G$ and hence either $(x,u)$ or $(u,v)$ is in $M'$, a contradiction. Thus, we have proved that if $e \in M\setminus M'$, then $e \notin E'$ and therefore, $M\cap E' = M'\cap E'$.
	
	We will now show that $M$ is a matching. From the above arguments, we know that for each edge $(u,v)=e \in M\setminus M'$, degree of both $u$ and $v$ are at most~$2$. Thus, the only edges adjacent with $e$ are the ones that form the induced copy of $P_4$ with it. As neither of them are in $M$, the edge $e$ does not share an endpoint with any other edge of $M$. This completes the proof of the claim.
\end{proof}

It remains to verify that $M$ is a feasible matching. First, suppose for the sake of contradiction that $G-M$ is not $P_4$-free. Fix an induced $P_4$ in $G-M$, say $H$, with edges $\{(u,v),(v,w),(w,x)\}$. The set $V(H) = \{u,v,w,x\}$ does not induce a $P_4$ in $G$, for otherwise, we would have added one of the edges of $H$ to $M$. Thus, we have, $(u,x) \in E$, implying that, $(u,x) \in M$. Recall that since $G$ is connected and of maximum degree at least $3$, the set $V(H)$ contains a vertex of degree at least $3$ in $G$, which implies that all edges of the $C_4$ induced by $V(H)$ must be in $E'$. Since $M\cap E' = M'\cap E'$, we have in particular $(u,x)\in M'$. Since $(u,x)$ is the only $M'$-edge in the $C_4$ induced by $V(H)$ in $G$, it contradicts the fact that the type (2) clause corresponding to that $C_4$ is satisfied by the assignment $a$.

% Then, $V(H)$ does not induce a $P_4$ in $G$ for otherwise its middle edge would have been added to $M$. Therefore, $V(H)$ induces a $C_4$ in $G$. Denoting by $e$ the edge completing the $P_4$ $H$ to a $C_4$ in $G$, we have that $e\in M$.
%Recall that since $G$ is connected and of maximum degree at least $3$,
%the set $V(H)$ contains a vertex of degree at least $3$ in $G$, which implies that all edges of the $C_4$ must be in $E'$.
%Since $M\cap E' = M'\cap E'$, we have in particular $e\in M'$, but since $e$ is the only edge from $M'$ in the $C_4$, this contradicts the fact that
%assignment $a$ satisfies the clause of type (2) corresponding to the $C_4$ and a $2$-matching in $C_4$ containing $e$.

Second, let $H$ be an induced subgraph of $G$ isomorphic to $C_6$. By the definition of $E'$, we have that $E(H)\subseteq E'$ and consequently
$M\cap E(H) = M'\cap E(H)$. The fact that the clauses of type (3) corresponding to $H$ are satisfied by $a$ implies that $H - M$ is a union of three disjoint edges. %$|M'\cap E(H)| = 3$ and hence $|M\cap E(H)| = 3$.

Finally, let $H$ be an induced subgraph of $G$ isomorphic to the domino (with edges labeled as in the right side of Figure~\ref{fig:C6-domino-fork}).
By the definition of $E'$, we again have $E(H)\subseteq E'$ and thus $M\cap E(H) = M'\cap E(H)$. The fact that the clauses of type (4) corresponding to $H$ are satisfied by $a$ implies that $M'\cap \{e_2,e_3\}\neq \emptyset$ and $M'\cap \{e_5,e_6\}\neq \emptyset$.
We may assume by symmetry that $M'\cap \{e_2,e_3\} = \{e_2\}$.
The fact that the clause of type (2) is satisfied corresponding to the $C_4$ with edge set $\{e_1,e_2,e_6,e_7\}$
and the $2$-matching $\{e_2,e_6\}$ implies that $e_6\in M'$.
Consequently, $M\cap E(H) = M'\cap E(H) = \{e_2,e_6\}$ and the desired condition holds.
This proves that $M$ is a feasible matching in $G$ and completes the proof of the lemma.
\end{proof}
%The proof of the above lemma is in the appendix. At a high level, the proof shows that every satisfying assignment to $\varphi$ corresponds to a matching in $G$ that is almost feasible, in the sense that we will be able to extend it (in polynomial time) to a feasible matching of $G$. Conversely, if $M$ is a feasible matching in $G$, the assignment that sets $x_e$ to true iff $e \in M$ is a satisfying assignment to $\varphi$.

The correctness of the algorithm follows from Theorem~\ref{thm:readability-2-characterization}.
%\pnote{instead of:from Lemmas~\ref{lem:induced-subgraph}{(b)} and~\ref{lem:induced-subgraph}{(d)}, Corollary~\ref{cor:path-cycle}, Theorem~\ref{thm:readability-2-characterization}, and Lemma~\ref{lem:2SAT}.}
We can compute formula $\varphi$ from a given graph $G$ in polynomial time. The $2$SAT problem is solvable in linear time~\cite{MR526451}, and clearly, all the other steps of the algorithm can be implemented to run in polynomial time.
The method given above can easily be modified so that it also efficiently computes an overlap labeling of length at most~$2$ in case of a yes instance.
%\nvnote{Commented out the remark that was below this. That is a minor point and I don't think we should spend a paragraph for it.}
\end{proof}

\section{Readability of grids and grid graphs}\label{sec:grids}
In this section, we determine the readability of grids by proving Theorem~\ref{thm:rgrid}.
%In this section, we determine the readability of two-dimensional grid graphs by proving Theorem~\ref{thm:rgrid}.
%Recall that a {\it (two-dimensional) grid graph} $G_{m,n}$ has vertex set $\{0,1,\dots,m-1\}\times \{0,1,\dots,n-1\}$, and two vertices $u$ and $v$ are adjacent in $G_{m,n}$  if and only if the $L_1$ distance between $u$ and $v$ is 1. %$\|u-v\|_1=1$,
%\nithin{where $\|u-v\|_1$ denotes the $L_1$ distance between $u$ and $v$.}\mnote{Is notation $\|u-v\|_1$ standard enough that it does not need to be defined?}\nvnote{Re to Martin: I think so. I have however added a phrase to make it clear.}
%
We first look at toroidal grids, which are closely related to grids.
For positive integers $m \geq 3$ and $n \geq 3$, the \textit{toroidal grid} ${\it TG}_{m,n}$ is obtained from the grid $G_{m,n}$ by adding edges $((i,0),(i,n-1))$ and $((0,j),(m-1,j))$ for all $i\in \{0,\dots,m-1\}$ and $j\in\{0,\dots,n-1\}$. (See Figure~\ref{fig:grid_graph_4x4} for an example.)
	The graph ${\it TG}_{m,n}$ is bipartite if and only if $m$ and $n$ are both even.
	%Indeed, if one of $m$ and $n$ is odd then ${\it TG}_{m,n}$ contains an odd cycle and is therefore not bipartite.
	In this case, a bipartition can be obtained by setting $V({\it TG}_{m,n}) = V_s \cup V_p$ where $V_s = \{(i,j)\in V({\it TG}_{m,n}): i+j \equiv 0\pmod 2\}$ and
	$V_p = \{(i,j)\in V({\it TG}_{m,n}): i+j \equiv 1\pmod 2\}$.
%We first determine the readability of toroidal grid graphs.

\begin{lemma}
	\label{lem:upper-bound-TG}
	For all integers $n>0$, we have $r({\it TG}_{4n,4n})\le 3$.
\end{lemma}
\begin{proof}
	Fix $n$ and let $G={\it TG}_{4n,4n}$. Each vertex $u$ of $G$ has associated coordinates $(u_1,u_2)$, where $u_1,u_2\in\{0,1,\dots,4n-1\}$. All arithmetic on coordinates will be performed modulo $4n$. By Lemma~\ref{lem:induced-subgraph}{(c)}, we may assume without loss of generality the bipartition $(V_s,V_p)$ given above.
	
	%The following sentence started as 'First, we decompose..'
	We decompose $G$ into three subgraphs. The first subgraph consists of squares. A {\em square} $S_u$ for a vertex $u$ of $G$ is the subgraph of $G$ induced by vertices $\{u, u+(0,1), u+(1,0),u+(1,1)\}$. The subgraph $G_0$ of $G$ is the union of all squares $S_u$, where either
	(1) $u_1$ is divisible by 4 and $u_2$ is divisible by 2, or
	(2) $u_1 + 2$ is divisible by 4 and $u_2 + 1$ is divisible by 2.
	\begin{figure}[t]
		\begin{center}
			\includegraphics[width=0.7\linewidth]{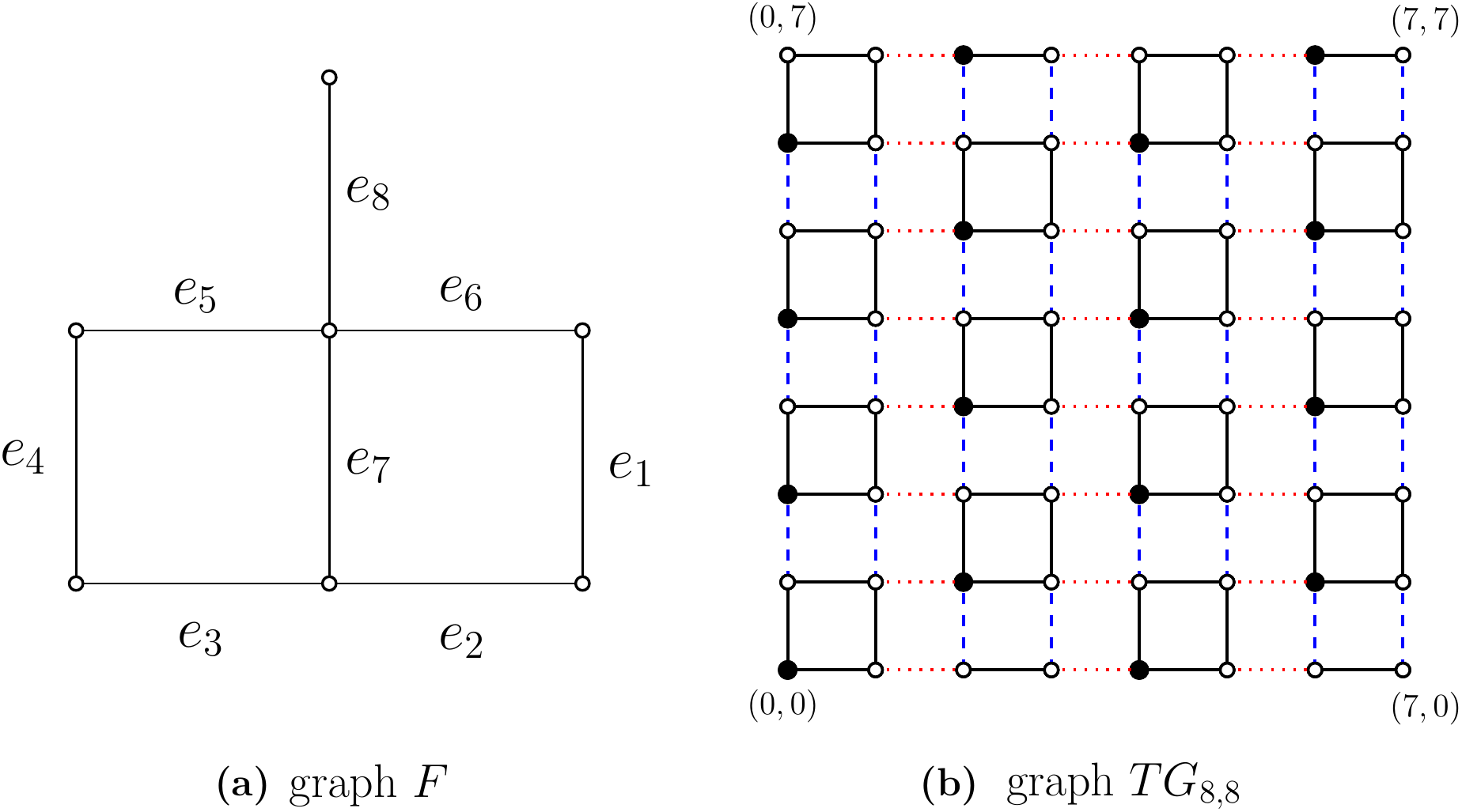}
		\end{center}
		\caption{{(a)} The graph $F$.   {(b)} The graph $TG_{8,8}$ (with edges between the extreme layers omitted in the drawing) where the lower left endpoints of the squares in $G_0$ are marked using solid dots, and the edges in $M_1$ and $M_2$ are drawn using (red) dotted and (blue) dashed lines, respectively.} \label{fig:squares_in_G0}
	\end{figure}

	We assign each square a unique identifier from the range $\{0,1,\dots, 4n^2 - 1\}$.
	Observe that each vertex $u$ of $G$ belongs to exactly one square in $G_0$,
	and we use $\ell_0(u)$ to denote the identifier of the square in $G_0$ to which $u$ belongs.
	We divide the edges of $G$ into {\it horizontal} and {\it vertical} ones respectively, according to whether they connect a pair of vertices that differ in their first, resp., second coordinates.
	Next, we define $M_1$ (respectively, $M_2$) to be the set of all horizontal (respectively, vertical) edges of $G - E(G_0)$.
	For $i\in\{1,2\}$, we use $M_i(u)$ to denote the vertex matched to $u$ in $M_i$. Figure~\ref{fig:squares_in_G0} (b) illustrates the graph $TG_{8,8}$ (without the dotted edges, for simplicity) where the lower left endpoints of the squares in $G_0$ are marked using black dots, and the edges in $M_1$ and $M_2$ are drawn using dotted and dashed lines, respectively. %colored red and blue.
	
	%\begin{definition}[Labeling of $G$]\label{def:labeling-of-toroidal-graph}
	We now define a labeling $\ell$ of $G$.
	For each vertex $u$ of $G$, define $\ell_1(u)=\ell_0(M_1(u))$ and $\ell_2(u)=\ell_0(M_2(u)).$
	If $u\in V_s$ then define $\ell(u)=\ell_2(u)\ell_1(u)\ell_0(u)$; if $u\in V_p$ then define $\ell(u)=\ell_0(u)\ell_1(u)\ell_2(u),$ i.e., the same characters, but in reverse order.
	%\end{definition}
	The following claim shows that $\ell$ is an overlap labeling and,
	since $\ell$ has labels of length 3,
	proves the lemma.
	\begin{claim}\label{claim:labeling-of-toroidal-graph}
		Labeling $\ell$ is an overlap labeling of $G$.
	\end{claim}
	\begin{proof}
		First, we make three observations about $G_0$, $M_1$, and $M_2$.
		\begin{observation}\label{obs:perfect-match}
			Each edge of $G$ is in exactly one of $E(G_0), M_1,$ and $M_2$.
			Both $M_1$ and $M_2$ are perfect matchings in $G$.
		\end{observation}
		
		\begin{observation}\label{obs:offset-between-columns}
			If $(u,v)\in M_2$ then $\ell_0(M_1(u))= \ell_0(M_1(v)),$
			i.e., $M_1$ matches $u$ and $v$ to vertices in the same square of $G_0$.
		\end{observation}
		\begin{observation}\label{obs:unique-edge-ids}
			%If $u$ and $v$ are distinct vertices of the same square in $G_0$ then $\ell_0(M_1(u))\neq \ell_0(M_1(v)).$
			Any pair of squares in $G_0$ is connected by at most~one edge in $M_1$. That is, for all pairs $(id_1,id_2)$ of square ids, at most~one edge $(u,v)\in M_1$ satisfies $\ell_0(u)=id_1$ and $\ell_0(v)=id_2$.
		\end{observation}

		First, we show that, for every edge $(u,v)$ of $G$, where $u\in V_s$ and $v\in V_p$, the label $\ell(u)$ overlaps the label $\ell(v)$. By Observation~\ref{obs:perfect-match}, each edge of $G$ is in one of $G_0, M_1$ and $M_2$. If $(u,v)$ is in $G_0$, then $u$ and $v$ belong to the same square of $G$ and, by construction $\ell_0(u)=\ell_0(v)$. That is,
		$$\suf_1(\ell(u))=\ell_0(u)=\ell_0(v)=pre_1(\ell(v)).$$
		%$\ov(\ell(u),\ell(v))=1$.
		%
		If $(u,v)\in M_1$, then $\ell_1(u)=\ell_0(v)$ and $\ell_1(v)=\ell_0(u),$ by the definition of $\ell_1$. Therefore, $$\suf_2(\ell(u))=\ell_1(u)\ell_0(u)=\ell_0(v)\ell_1(v)=\pre_2(\ell(v)).$$
		If $(u,v)\in M_2$, then $\ell_2(u)=\ell_0(v)$ and $\ell_2(v)=\ell_0(u),$ by the definition of $\ell_2$. By Observation~\ref{obs:offset-between-columns} and the definition of $\ell_1$, we get that $\ell_1(u)=\ell_1(v)$. That is, $$\ell(u)=\ell_2(u)\ell_1(u)\ell_0(u)=\ell_0(v)\ell_1(v)\ell_2(v)=\ell(v).$$
		In all three cases $\ell(u)$ overlaps $\ell(v)$.
		
		\medskip
		
		It remains to show that if, for $u\in V_s$ and $v\in V_p,$ label $\ell(u)$ overlaps label $\ell(v)$ then $(u,v)$ is an edge in $G$. Since labels $\ell(u)$ and $\ell(v)$ have length 3, the overlap from $\ell(u)$ to $\ell(v)$ can be of length 1, 2 or 3.
		If $\suf_1(\ell(u))=\pre_1(\ell(v))$ then $\ell_0(u)=\ell_0(v)$, that is, $u$ and $v$ are in the same square of $G_0$. Hence, $(u,v)$ is an edge in $G_0$ and, consequently, in $G$.
		
		If $\suf_2(\ell(u))=\pre_2(\ell(v))$ then $\ell_1(u)=\ell_0(v)$ and $\ell_0(u)=\ell_1(v)$. By the definition of $\ell_1$, this implies that both $(u, M_1(u))$ and $(M_1(v),v)$ connect squares of $G_0$ with identifiers $\ell_0(u)$ and $\ell_0(v)$. By Observation~\ref{obs:unique-edge-ids}, $(u, M_1(u))$ is the same edge as $(M_1(v),v)$, namely, $(u,v)$. Hence, $(u,v)$ is in $M_1$ and, consequently, in $G$.
		
		Finally, suppose $\ell(u)=\ell(v)$. Then $\ell_2(u)\ell_1(u)\ell_0(u)=\ell_0(v)\ell_1(v)\ell_2(v).$ Since $\ell_2(u)=\ell_0(v)$ and $\ell_0(u)=\ell_2(v)$, it follows that $(u, M_2(u))$ and $(M_2(v),v)$ are vertical edges connecting the same pair of squares in $G_0$. Since $\ell_1(u)=\ell_1(v)$, we have that $M_1(u)$ and $M_1(v)$ belong to the same square in $G_0$. Both conditions can hold only if  $(u, M_2(u))$ and $(M_2(v),v)$ are the same edge, namely, $(u,v)$. Hence, $(u,v)$ is in $M_2$ and, consequently, in $G$.
		In all cases, we proved that $(u,v)$ is an edge of $G$. %This completes the proof of Claim~\ref{claim:labeling-of-toroidal-graph}.
	\end{proof}
	
	This completes the proof of Lemma~\ref{lem:upper-bound-TG}.
	\end{proof}

We can now prove Theorem~\ref{thm:rgrid}, about the readability of $G_{m,n}$.
We first recall the following simple observation (which follows, e.g., from~\cite[Theorem 4.3]{ChikhiMMR16}).
%By $P_4$ we denote the path on $4$ vertices.

\begin{lemma}\label{lem:small-readability}
	A bipartite graph $G$ has: (i) $r(G) = 0$ if and only if $G$ is edgeless, and
	(ii) $r(G) \le 1$ if and only if $G$ is $P_4$-free (equivalently: a disjoint union of bicliques).
\end{lemma}

\begin{proof}[Proof of Theorem~\ref{thm:rgrid}]%[Proof for the case of $m< 3$ or $n < 3$]
	First, by Lemma~\ref{lem:small-readability}, $r(G_{m,n})$ is $0$ if $m=n=1$ and positive, otherwise.
	Second, when $(m,n)\in \{(1,2),(1,3),(2,2)\}$, the graphs $G_{m,n}$ are isomorphic to $K_{1,1}, K_{1,2},$ and $K_{2,2}$, respectively. Thus, by Lemma~\ref{lem:small-readability}, their readability is~$1$.
	
	Third, when $m+n\geq 5$, the grid $G_{m,n}$ contains an induced $P_4$, implying that $r(G_{m,n})\geq 2.$ By
	Theorem~\ref{thm:readability-2-characterization}, a twin-free bipartite graph $G$ has readability at most~2 if and only if
	$G$ has a feasible matching. (See Definition~\ref{def:feasible-matching}.) When $m+n\geq 5$, the grid $G_{m,n}$ is twin-free. If $m=2$ and $n\ge 3$, then $M=\{((i,j),(i,j+1))\mid i\in\{0,1\}$ and $j\in\{0,\dots,n-2\}$ is even$\}$ is a feasible matching in $G_{m,n}$, so $r(G_{m,n})= 2$.
	If $m = 1$ and $n\ge 4$, then $G_{m,n}$ is isomorphic to a path of length at least three. Since its maximum degree is 2, we get $r(G_{m,n})\leq 2$, by Corollary~\ref{cor:path-cycle}. Thus, $r(G_{m,n})= 2$.
%\end{proof}
%\begin{proof}[Proof for the case of $m\geq 3$ and $n\geq 3$]
%	For the cases of $m<3$ or $n<3$, see the appendix.
	
	To show that $r(G_{m,n})\leq 3$ for $m\geq 3$ and $n\geq 3$,
	we observe that $G_{m,n}$ (for $m\leq n$) is an induced subgraph of ${\it TG}_{4n,4n}$.
	By Lemmas~\ref{lem:induced-subgraph}{(a)} and \ref{lem:upper-bound-TG},
	we have that $r(G_{m,n}) \leq r({\it TG}_{4n,4n})\leq 3$.
	
	To show that $r(G_{m,n})\ge 3$,
	let $F$ be the graph obtained by taking the graph $G_{3,2}$ and adding a new vertex adjacent to one of the degree-$3$ vertices of $G_{3,2}$; see Figure~\ref{fig:squares_in_G0}(a). Clearly, $F$ is a bipartite graph and an induced subgraph of $G_{m,n}$.
	Since $F$ is also twin-free, we can prove that $r(F)>2$ by applying
	Theorem~\ref{thm:readability-2-characterization}, provided we show that $F$ does not have a feasible matching.
	Assume the edge labeling as in Figure~\ref{fig:squares_in_G0}(a) and suppose for a contradiction that $F$ has a feasible matching $M$. The third condition in Definition~\ref{def:feasible-matching} implies that $M\cap (E(F)\setminus \{e_8\}) \in \{\{e_2,e_6\},\{e_3,e_5\}\}$. By symmetry, we may assume that $M\cap (E(F)\setminus \{e_8\}) = \{e_2,e_6\}$.
	Since $M$ is a matching, we have $e_8\not\in M$. But now the graph $F-M$ contains an induced $P_4$ with edge set $\{e_4,e_5,e_8\}$, a contradiction to the fact that $M$ is feasible. This shows that $r(F)\ge 3$.
	%
	%If $m\ge 3$ and $n\ge 3$, then $G_{m,n}$ contains an induced copy of $F$, for example, on the vertex set $\{(0,0), (1,0), (2,0), (0,1), (1,1), (2,1),  (1,2)\}$.
	By Lemma~\ref{lem:induced-subgraph}{(a)},
	$r(G_{m,n})\ge r(F) \ge 3$ if $m\ge 3$ and $n\ge 3$.
\end{proof}

\section{Conclusion}

%\textbf{[TODO:] We should have some form of conclusion and perhaps open problems or directions of future work?} %SK

In this work we gave several results on families of $n$-vertex bipartite graphs with readability $o(n)$. The results were obtained by developing new or applying a variety of known techniques to the study of readability. 
These include a graph theoretic characterization in terms of matchings, a reduction to 2SAT, an explicit construction of overlap labelings analyzed via number theoretic notions, and a new lower bound applicable to dense graphs with a large number of distinct degrees. One of the main specific questions left open by our work is to close the gap between the $\Omega(\log n)$ lower bound and the $\Oh(\sqrt n)$ upper bound on the readability of $n$-vertex bipartite chain graphs.
In the context of general bipartite graphs, it would be interesting to determine the computational complexity of determining whether the readability of a given bipartite graph is at most~$k$, where $k$ is either part of input or a constant greater than $2$, to study the parameter from an approximation point of view, and to relate it to other graph invariants. For instance, for a positive integer $k$, what is the maximum possible readability of a bipartite graph of maximum degree at most~$k$? Another interesting direction would be to study the complexity of various computational problems on graphs of low readability.

\paragraph{Acknowledgments.}

The result of Section~\ref{subsec:cnn-upperbound} was discovered with the help of
	The On-Line Encyclopedia of Integer Sequences~\textregistered{}~\cite{OEIS}.
This work has been supported in part by NSF awards DBI-$1356529$, CCF-$1439057$, IIS-$1453527$, and IIS-$1421908$ to P.M. and by the Slovenian Research Agency (I$0$-$0035$, research program P$1$-$0285$ and research projects N$1$-$0032$, J$1$-$6720$, and J$1$-$7051$) to M.M. The authors S.R. and N.V. were supported in part by NSF grant CCF-1422975 to S.R. The author N.V. was also supported by Pennsylvania State University College of Engineering Fellowship and Pennsylvania State University Graduate Fellowship and in part by NSF grant IIS-$1453527$ to P.M. %Part of the work was carried out during a visit of M.M.~to S.K.~at the University of Bonn; their hospitality is gratefully acknowledged.
V.J.~did most of his work on the paper while he was an undergraduate student at the University of Primorska. The main idea of the proof of Lemma~\ref{lem:upper-bound-TG} was developed in his final project paper~\cite{Jovicic}.
   % \pnote{I tried to resolve the following discussion by moving the sentence to the acknowledgments and slight rephrasing it \sofya{SR: I do not think that this sentence is necessary.\martin{MM: It certainly does not fit here. The main idea of Vladan led to the statement of Lemma~\ref{lem:upper-bound-TG} so it would better fit there. We can also remove it. I thought it would be fair to Vladan to acknowledge that the idea of the proof of Lemma~\ref{lem:upper-bound-TG} appeared in his final project paper. But since Vladan is a coauthor of the present paper, I suppose it is fine to remove it. I am not sure what is the standard about these things. \rayan{RC: if this is to keep citation [6], it could also appear in the very last sentence of the acknowledgements.}}}}

\bibliography{references}

\end{document}